\documentclass[11pt]{article}%
\pdfoutput=1
\usepackage[hidelinks]{hyperref}
\usepackage{bbm,url,microtype,amssymb,amsthm,mathtools,mathrsfs,graphicx,float,color}
\usepackage{fullpage}
\usepackage{lmodern}
\usepackage{dsfont}
\usepackage[T1]{fontenc}
\usepackage[USenglish]{babel}
\usepackage{enumerate}
\usepackage{booktabs}
\usepackage{subcaption}
\usepackage{tikz}
\usetikzlibrary{arrows.meta,calc,decorations.pathmorphing}
\usepackage{authblk}
\usepackage{bm}

\providecommand{\U}[1]{\protect\rule{.1in}{.1in}}
\newtheorem{thm}{Theorem}
\newtheorem{lem}[thm]{Lemma}
\newtheorem{prp}[thm]{Proposition}
\newtheorem{cor}[thm]{Corollary}

\newtheorem{dfn}[thm]{Definition}

\usepackage{cleveref}
\crefname{thm}{Theorem}{Theorems}
\Crefname{thm}{Theorem}{Theorems}
\crefname{lem}{Lemma}{Lemmas}
\Crefname{lem}{Lemma}{Lemmas}
\crefname{prp}{Proposition}{Propositions}
\Crefname{prp}{Proposition}{Propositions}
\crefname{cor}{Corollary}{Corollaries}
\Crefname{cor}{Corollary}{Corollaries}
\crefname{prb}{Problem}{Problems}
\Crefname{prb}{Problem}{Problems}
\crefname{dfn}{Definition}{Definitions}
\Crefname{dfn}{Definition}{Definitions}
\crefname{obs}{Observation}{Observations}
\Crefname{obs}{Observation}{Observations}
\crefname{section}{Section}{Sections}
\crefname{appendix}{Appendix}{Appendices}
\numberwithin{equation}{section}
\let\oldref\ref
\renewcommand{\ref}[1]{(\oldref{#1})}

\newcommand{\R}{\mathbb{R}}

\newcommand{\ket}[1]{\vert #1\rangle }
\newcommand{\bra}[1]{\langle #1 \vert}

\begin{document}
\title{Violating a Bell Inequality:\\What to Do When You Lose Your Quantum System}
\author[1,2,3]{Xinyu Xu}
\author[4]{Yuqing Li}
\author[5]{Mingze Xu}
\author[6,1,3,7]{Dawei Ding\thanks{daweiding@fudan.edu.cn}}
\affil[1]{\small Shanghai Institute for Mathematics and Interdisciplinary Sciences (SIMIS), Shanghai 200433, China}
\affil[2]{\small Research Institute of Intelligent Complex Systems, Fudan University, Shanghai 200433, China}
\affil[3]{\small Yau Mathematical Sciences Center, Tsinghua University, Beijing 100084, China}
\affil[4]{\small Department of Computer Science, University of Pittsburgh, Pittsburgh, Pennsylvania 15260, USA}
\affil[5]{\small Department of Electrical and Computer Engineering, University of Illinois Urbana-Champaign, Urbana, Illinois 61801, USA}
\affil[6]{\small Center for Mathematics and Interdisciplinary Sciences, Fudan University, Shanghai 200433, China}
\affil[7]{\small Beijing Institute of Mathematical Sciences and Applications (BIMSA), Beijing 101408, China}
\date{}
\maketitle

\begin{abstract}
The detection loophole is a well-known loophole in Bell nonlocality motivated by the inefficiency of detectors used in Bell experiments. However, this inefficiency is actually a ubiquitous feature of any photonic platform for quantum networks since photons are readily absorbed by the surrounding environment. When the quantum system, usually a photon, is lost, what should the parties output? The most natural choice for parties experiencing loss is to fall back to a deterministic strategy where each party produces an output via a function of her local input. In this paper, we study how to choose such a fallback strategy for general Bell inequalities with respect to different efficiencies $\eta$ to maximize the possible violation. We mathematically prove that for the CHSH inequality, there exists a fallback strategy that is optimal for all $\eta$. However, in general, we find that the optimal fallback strategy can vary with $\eta$, sometimes in surprising ways. For example, we find an example of a Bell inequality where the optimal fallback strategy near the threshold efficiency is not an optimal deterministic strategy in the lossless setting. Our results can reduce the efficiency requirements for closing the detection loophole and thus are useful for applications of Bell inequality violation, such as quantum telepathy or device-independent quantum key distribution.
\end{abstract}

\tableofcontents

\section{Introduction}
\label{sec:intro}

Quantum nonlocality is the phenomenon in which multiple parties can share correlations using quantum entanglement that are impossible to produce classically. This is also known as Bell inequality violation~\cite{bell1964einstein} and is one of the most counterintuitive properties of quantum mechanics. Early experiments have demonstrated this provable separation between classical and quantum mechanics~\cite{freedman1972experimental,fry1976experimental,aspect1982locality,weihs1998violation}, ultimately culminating in the 2022 Nobel prize in physics~\cite{aspect2022nobel}. 

However, to convincingly demonstrate the violation of a Bell inequality, experiments need to close a series of possible loopholes that still allow the possibility of a classical explanation. One well-known loophole is the \textbf{detection loophole}~\cite{PhysRevA.47.R747}. This historically referred to the problem of inefficient detectors used in experiments that sometimes fail to measure the quantum system. Such a setting is mathematically modeled with a parameter called the \textbf{efficiency} $\eta \in [0,1]$, which is defined as the probability of successfully completing the quantum measurement. 
This type of inefficiency is actually quite prevalent in physical implementations of  quantum networks.
%More generally, other physical effects can also be mathematically modeled by this inefficiency. 
For example, in experiments that directly measure pairs of entangled photons from a common source, the photons may be absorbed by the transmission medium before reaching the detectors. In this case the efficiency decays exponentially with the distance between the detectors and the source. \emph{In general, loss is an unavoidable problem for any photonic platform, the physical platform of choice for quantum communication networks.}
Even for memory-based implementations, an inefficiency may exist. For example, if neutral atoms are used as quantum memories~\cite{covey2023quantum}, the atoms can be lost to the environment due to an imperfect vacuum. 
Furthermore, measurements of atomic systems often involve measuring fluorescence, which can be plagued by photon collection problems or detector inefficiency.
These nonidealities can also be mathematically modeled by an $\eta$ parameter. In general, these physical effects all involve a \textbf{loss} of the quantum system, whether it be a photon or an atom. For linguistic convenience, we will use such terminology going forward. The loss of a quantum system is depicted schematically in~\Cref{fig:lossy_scenario}.
\begin{figure} [h]
    \centering
    \definecolor{accent}{RGB}{42,82,126}
    \definecolor{loss}{RGB}{174,49,50}
    \definecolor{linegray}{RGB}{65,69,74}
        \begin{tikzpicture}[
          font=\small, >=Latex,
          source/.style={circle,draw=accent,line width=.9pt,fill=accent!7,minimum size=16mm,align=center},
          party/.style={rounded corners=1.5pt,draw=linegray,line width=.8pt,fill=black!2,minimum width=18mm,minimum height=9mm,align=center},
          junction/.style={circle,draw=linegray,fill=white,line width=.75pt,minimum size=3.1mm,inner sep=0pt},
          photon/.style={draw=accent,line width=.85pt,decorate,decoration={snake,amplitude=.55mm,segment length=3.2mm}},
          received/.style={draw=linegray,line width=.8pt,-{Latex[length=2.1mm,width=1.5mm]}},
          lost/.style={draw=loss,line width=.8pt,dashed,-{Latex[length=1.9mm,width=1.4mm]}},
          label/.style={fill=white,inner sep=1.2pt},
          lossprob/.style={anchor=west,inner sep=0pt,xshift=4pt}]
        
        \node[source] (S) at (0,0) {$\rho$};
        \node[above=2pt] at (S.north) {entangled state};
        \node[party] (P1) at (6.8,3.2) {Party $1$};
        \node[party] (P2) at (6.8,0) {Party $2$};
        \node at (6.8,-1.45) {$\vdots$};
        \node[party] (Pn) at (6.8,-3.2) {Party $n$};
        
        \coordinate (L1) at ($(S)!0.58!(P1)$);
        \coordinate (L2) at ($(S)!0.58!(P2)$);
        \coordinate (Ln) at ($(S)!0.58!(Pn)$);
        \draw[photon] (S) -- (L1);
        \draw[photon] (S) -- (L2);
        \draw[photon] (S) -- (Ln);
        \node[junction] at (L1) {};
        \node[junction] at (L2) {};
        \node[junction] at (Ln) {};
        
        \draw[received] (L1) -- node[pos=.76,above=5pt,inner sep=0pt] {$\eta_1$} (P1);
        \draw[received] (L2) -- node[pos=.62,above=5pt,inner sep=0pt] {$\eta_2$} (P2);
        \draw[received] (Ln) -- node[pos=.76,above=5pt,inner sep=0pt] {$\eta_n$} (Pn);
        \draw[lost] (L1) -- ++(0,-1.08)
          node[pos=.62,lossprob] {$1-\eta_1$}
          node[below,text=loss,inner sep=1pt] {lost};
        \draw[lost] (L2) -- ++(0,-1.08)
          node[pos=.62,lossprob] {$1-\eta_2$}
          node[below,text=loss,inner sep=1pt] {lost};
        \draw[lost] (Ln) -- ++(0,-1.08)
          node[pos=.62,lossprob] {$1-\eta_n$}
          node[below,text=loss,inner sep=1pt] {lost};
    \end{tikzpicture}
    \caption{An entangled quantum state, with each subsystem independently undergoing loss. }
    \label{fig:lossy_scenario}
\end{figure}
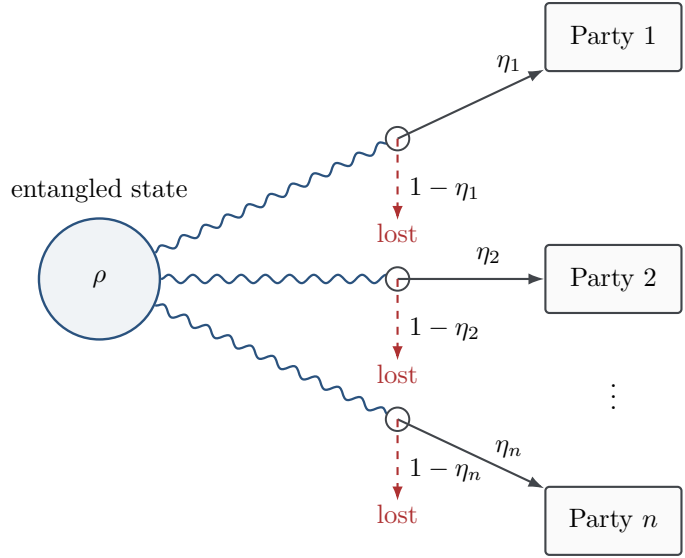

In the event of a loss, a party no longer has access to her share of a quantum system. However, she is still required to produce an output. This is intended to close the detection loophole, whereby a classical explanation remains possible through manipulation of the statistics of the measurement outcomes. Beyond foundational concerns, the parties are also required to produce outputs in applications of Bell inequality violation. 
For example, in quantum telepathy~\cite{ding2026quantum}, parties use quantum entanglement to coordinate decisions in a real-world scenario where communication between parties is restricted, such as in high-frequency trading~\cite{brandenburger2016team,szegedy2020systems,ding2024coordinating} or distributed systems~\cite{da2025entanglement,arun2025faster}. In such scenarios, the parties are assumed to always produce an output\footnote{Inaction can also be defined as an output. } in response to their local inputs. 
%The resulting coordinated behavior in the presence of loss should still be non-classical for a quantum advantage to persist.
This requirement is also necessary for an unconditional security guarantee in device-independent quantum key distribution (DIQKD)~\cite{zapatero2023advances}. Now, how can the party produce an output without a quantum system? A natural solution is to fall back to a classical strategy, where each party experiencing loss applies a function to her local input to produce her output~\cite{branciard2011detection,ding2024coordinating,gigena2024robustselftestingbellinequalities}. 

In this paper, we investigate the question of how to determine the optimal ``fallback strategy'' for a general Bell inequality in the presence of loss. Optimality here means the highest possible violation of the Bell inequality with respect to an efficiency $\eta$. 
Finding such strategies is important for closing the detection loophole at lower efficiencies. It is also crucial for achieving higher performance in applications of Bell inequality violation, such as quantum telepathy and DIQKD, if the physical quantum network used exhibits loss.
Now, it is easy to see that to find such optimal fallback strategies, it is sufficient to consider deterministic strategies where each party's output is a deterministic function of her input (adding randomness leads to taking convex combinations which does not increase the value of the Bell expression). However, such deterministic strategies are discrete objects, while the efficiency $\eta \in [0,1]$ is a continuous parameter. Hence, we expect the set of optimal fallback strategies to suddenly change as we vary $\eta$. For some Bell inequalities such as the CHSH inequality, the same fallback strategy is optimal for all $\eta \in [0,1]$. In particular, the fallback strategy is optimal for $\eta=0$ at which point the quantum system can never be used. Hence, it is also the optimal deterministic strategy for the CHSH inequality without loss. This universality property for the CHSH inequality was shown numerically in~\cite{gigena2024robustselftestingbellinequalities}, while we prove it analytically below. In contrast, Ref.~\cite{PhysRevA.97.062123} numerically observed for the $I_{3322}$ inequality, within a restricted class of two-qubit strategies, that different fallback strategies can become optimal at different detection efficiencies.

Our paper is structured as follows. In~\Cref{sec:lossy_behaviors}, we use the language of nonlocal games to make mathematical definitions that describe how loss can affect the parties' behavior (input-output distribution). We show in particular that the effect of loss and the choice of a fallback strategy can be incorporated via a transformation of the corresponding Bell expression. The transformed Bell expression can subsequently be optimized over all quantum strategies like any other Bell expression using numerical techniques such as the see-saw method~\cite{werner2001bell,Ketjl} or the Navascues-Pironio-Acín (NPA) hierarchy~\cite{navascues2008convergent}. In~\Cref{sec:fallback}, we prove some general results about optimal fallback strategies, including a robustness theorem giving conditions for which an optimal fallback strategy remains optimal for $\eta \in (1-\varepsilon,1)$ for some $\varepsilon >0$. In~\Cref{sec:case_studies}, we conduct case studies of well-known Bell inequalities such as the CHSH inequality and the $I_{3322}$ inequality in the presence of loss. In particular, given a Bell inequality we analytically or numerically determine whether or not there exists a single fallback strategy that is optimal for all $\eta \in [0,1]$. Surprisingly, we also find that the optimal fallback strategy at $\eta$ values near the threshold efficiency can be a non-optimal deterministic strategy that does not attain the classical value in the lossless setting.
We conclude in~\Cref{sec:conc}.

\section{Lossy Behaviors in Nonlocal Games}
\label{sec:lossy_behaviors}
\subsection{Nonlocal game}
We will use the language of nonlocal games, which is an equivalent formulation of Bell inequalities. In our paper, we will use these two terms interchangeably.
We first review the definition of a nonlocal game. Define $[n] \coloneqq \{1, 2, \cdots n\}$.
\begin{dfn}[Nonlocal game]
Let $n\ge 2$ be an integer. For each $i\in[n]$, let $S_i$ and $A_i$ be finite sets. We call $S_i$ the \textbf{input set} and $A_i$ the \textbf{output set} of party $i$. Define
\begin{equation*}
    S \coloneqq \prod_{i=1}^n S_i,
    \qquad
    A \coloneqq \prod_{i=1}^n A_i.
\end{equation*}
We define the \textbf{utility function} as a function
\begin{equation*}
    \mathcal V:A \times S\rightarrow \R
\end{equation*}
Let $\pi$ be a probability distribution over $S$, which we call the \textbf{input distribution}. A \textbf{nonlocal game} is specified by the tuple $G = (\mathcal V,\pi)$.
\end{dfn}

In a nonlocal game with $n$ parties, each party $i$ receives an input $s_i\in S_i$ and produces an output $a_i\in A_i$. The parties are assumed to know the input distribution $\pi$ and the utility function $\mathcal V$ in advance. Before the game starts, they may agree on a joint strategy to produce outputs given their inputs in order to maximize the expected utility, but no communication is allowed after the inputs are received.

Executing a strategy induces a \textbf{behavior}, namely a conditional probability distribution $p(\mathbf{a}|\mathbf{s})$,
where $\mathbf s=(s_1,\ldots,s_n)$ and $\mathbf a=(a_1,\ldots,a_n)$. 
This behavior describes the distribution of the outputs conditioned on the inputs. The expected utility 
%\textbf{winning probability} 
associated with a behavior $p(\mathbf a|\mathbf s)$ is defined as
\begin{equation*}
    \omega(G,p)
    %p_{\mathrm{win}}
    \coloneqq
    \sum_{\mathbf s\in S}
    \sum_{\mathbf a\in A}
    \pi(\mathbf{s})\,
    p(\mathbf{a}|\mathbf{s})\,
    \mathcal V(\mathbf{a},\mathbf{s}).
\end{equation*}
We also define the expected utility of a strategy $\mathcal S$ as 
$$\omega(G, \mathcal S) \coloneqq \omega(G,p),$$
where $p$ is the behavior induced by $\mathcal S$. We will sometimes suppress $G$ if it is clear from context.

We next describe the behaviors that can be generated using classical resources. We begin with deterministic strategies.

\begin{dfn}[Deterministic strategy]
Let $G=(\mathcal V,\pi)$ be a nonlocal game. A \textbf{deterministic strategy} $\mathcal S_d$ is specified by a collection of functions $\{f_i\}_{i=1}^n$, where
\begin{equation*}
    f_i:S_i\rightarrow A_i.
\end{equation*}
The induced \textbf{deterministic behavior} is
\begin{equation*}
    p_d(\mathbf a|\mathbf s)
    =
    \prod_{i=1}^n \delta_{a_i,f_i(s_i)},
\end{equation*}
where $\mathbf a\in\prod_{i=1}^n A_i$ and $\mathbf s\in\prod_{i=1}^n S_i$.
\end{dfn}

A general \textbf{classical strategy} $\mathcal S_c$ allows the parties to use shared randomness. Equivalently, every \textbf{classical behavior} can be written as a convex combination of deterministic behaviors:
\begin{equation*}
%\label{eq:classical_strategy}
    p_c(\mathbf a|\mathbf s)
    =
    \sum_{\lambda} q(\lambda)
    \prod_{i=1}^n \delta_{a_i,f_i^\lambda(s_i)},
\end{equation*}
where $\lambda$ denotes the shared random variable, $q(\lambda)$ is its probability distribution, and each $f_i^\lambda$ is a function $f_i^\lambda:S_i\to A_i$. The \textbf{classical value} of the game, denoted by $\omega_c(G)$, is the maximum of $\omega(G,p_c)$
% winning probability 
over all classical strategies. Since the expected score is linear in the behavior, this maximum can be attained by optimizing over deterministic strategies.

We now turn to quantum strategies. In a quantum strategy, the parties share a quantum state before the game starts and, after receiving their inputs, perform corresponding measurements on their respective subsystems.

\begin{dfn}[Quantum strategy]
Let $G=(\mathcal V,\pi)$ be a nonlocal game. For each $i\in[n]$, let $\mathcal H_i$ be a finite dimensional Hilbert space. Let $\ket{\psi}$ be a quantum state
\begin{equation*}
    \ket{\psi}\in \bigotimes_{i=1}^n \mathcal H_i.
\end{equation*}
For each input $s_i\in S_i$, let $M_i(s_i)$ be a projective measurement 
\begin{equation*}
    M_i(s_i)=\{E_{i,a_i}^{(s_i)}\}_{a_i\in A_i}
\end{equation*}
on $\mathcal H_i$ with results in $A_i$. A \textbf{quantum strategy} $\mathcal S_q$ is the tuple $(\ket{\psi}, \{M_i(s_i)\}_{i \in [n], s_i\in S_i})$.
The \textbf{quantum behavior} induced by $\mathcal S_q$ is
\begin{equation*}
    p_q(\mathbf a|\mathbf s)
    =
    \bra{\psi}
    \bigotimes_{i=1}^n
    E_{i,a_i}^{(s_i)}
    \ket{\psi}.
\end{equation*}
\end{dfn}

Every classical behavior is also a quantum behavior. Quantum behaviors
may achieve a higher expected utility for certain nonlocal games.
Let $\mathcal Q$ denote the set of behaviors realizable by
finite-dimensional quantum strategies, and let
$\overline{\mathcal Q}$ denote its closure in the finite-dimensional
Euclidean space of behaviors. In general, $\mathcal Q$ need not be
closed~\cite{slofstra2019set}. The \textbf{quantum value} of the game,
denoted by $\omega_q(G)$, is defined as
\begin{equation*}
    \omega_q(G)
    \coloneqq
    \sup_{p\in\mathcal Q}\omega(G,p)
    =
    \max_{p\in\overline{\mathcal Q}}\omega(G,p).
\end{equation*}

\subsection{Lossy behaviors}
We now define induced behaviors under loss.
Let $T\subseteq [n]$ denote the subset of parties who lose their quantum system. Parties in $T$ produce their outputs using a preset deterministic strategy $\mathcal{S}_d$, while parties in $[n]\setminus T$ perform quantum measurements on their received systems. We adapt the definition introduced in Ref.~\cite{ding2024coordinating} 
to describe the resulting behavior.
\begin{dfn}[Semideterministic strategy]
    Let $G=(\mathcal V,\pi)$ be a nonlocal game. 
    Given a quantum strategy $\mathcal S_q$ and a deterministic strategy $\mathcal S_d$ for $G$, the \textbf{semideterministic strategy}, denoted by the tuple $(\mathcal S_q, \mathcal S_d)$, is the strategy in which each party
    who do not suffer loss each performs the measurement specified in 
    $\mathcal S_q$, while each party who do suffer loss each implements the function specified in $\mathcal S_d$, which we call the \textbf{fallback strategy}.\footnote{Of course, this is also a deterministic strategy, but we introduce this new terminology to make it clear that this is in the context of loss. }
    Let $T \subseteq [n]$ denote the parties that suffer loss, and $\bar{T} = [n] \backslash T$ the parties that do not suffer loss.
    Their behavior can be described by %\XY{$p_{q,\bar{T}}(\mathbf{a}|\mathbf{s})$ or $p_{\mathcal S_q,\bar{T}}(\mathbf{a}|\mathbf{s})$ here?}
    \begin{equation*}
        p_{\mathcal S_q \sqcup_T \mathcal S_d}(\mathbf a \vert \mathbf s) \coloneqq \prod_{i\in T} \delta_{a_i,f_i(s_i)} \cdot p_{q,\bar{T}}(\mathbf{a}_{\bar T}|\mathbf{s}_{\bar T}),
    \end{equation*}
    where $f_i$ are the functions used in $\mathcal S_d$ and
    \begin{equation*}
        p_{q,\bar{T}}(\mathbf{a}_{\bar T}|\mathbf{s}_{\bar T}) \coloneqq \langle \psi |\bigotimes_{i\in T} I_{\mathcal H_i} \otimes  \bigotimes_{j\notin T} E_{j,a_j}^{(s_j)}|\psi \rangle.
    \end{equation*}
    is the behavior of the parties in $\bar T$.

    Let $\bm{\eta}$ be the efficiency parameters. For a semideterministic strategy $(\mathcal S_q,\mathcal S_d)$, the induced \textbf{$\bm \eta$-semideterministic behavior} is given by
    \begin{equation}
    \label{eq:lossy_behavior}
        p_{(\mathcal S_q , \mathcal S_d), \bm \eta}(\mathbf a\vert \mathbf s) \coloneqq \sum_{T \subseteq [n]} \prod_{i \in T} (1-\eta_i) \prod_{j \not \in T} \eta_j \cdot p_{\mathcal S_q \sqcup_T \mathcal S_d}(\mathbf a \vert \mathbf s).
    \end{equation}
\end{dfn}

By using shared randomness, the parties can realize a convex combination of the
behaviors in \Cref{eq:lossy_behavior}. We call such a strategy a \textbf{semiclassical strategy}. The induced \textbf{$\bm\eta$-semiclassical behavior} is given by
\begin{equation*}
p_{\bm\eta}(\mathbf a|\mathbf s)
=
\sum_{\lambda}
q(\lambda)
p_{(\mathcal S_q^\lambda,\mathcal S_d^\lambda),\bm\eta}
(\mathbf a|\mathbf s),
\end{equation*}
where $q(\lambda)$ is a probability distribution and
$(\mathcal S_q^\lambda,\mathcal S_d^\lambda)$ is a semideterministic strategy for each $\lambda$. We also refer to this as an \textbf{$\bm\eta$-lossy behavior}. 
Since the expected utility is linear in the behavior, using shared randomness does not increase the maximum achievable expected utility under loss.
The \textbf{$\bm \eta$-lossy value} of $G$ is thus defined as 
\begin{equation*}
\omega_{\bm\eta}
\coloneqq
\sup_{(\mathcal S_q,\mathcal S_d)}
\sum_{\mathbf s\in S}
\sum_{\mathbf a\in A}
\pi(\mathbf s)
p_{(\mathcal S_q,\mathcal S_d),\bm\eta}(\mathbf a|\mathbf s)
\mathcal V(\mathbf a,\mathbf s).
\end{equation*}
where the supremum is taken over all quantum strategies $\mathcal S_q$
and deterministic strategies $\mathcal S_d$. Since there are only finitely
many deterministic strategies, there exists a deterministic strategy
$\mathcal S_d^*$ for which this supremum is attained with respect to the
choice of fallback strategy. We call any such $\mathcal S_d^*$ an
\textbf{optimal fallback strategy at efficiency $\bm\eta$}.

\subsection{Computing the lossy value}
We wish to compute the lossy value of a nonlocal game. 
The expected utility of a nonlocal game $G=(\mathcal V,\pi)$ naturally
defines a \textbf{Bell expression} as follows:
\begin{equation*}
    \mathcal{B} = 
    \sum_{\mathbf{a},\mathbf{s}} \beta_{\mathbf{a}}^{\mathbf{s}} p(\mathbf{a}|\mathbf{s})
\end{equation*}
where 
\begin{equation*}
    \beta_{\mathbf{a}}^{\mathbf{s}} \coloneqq \pi(\mathbf{s}) \mathcal{V}(\mathbf{a},\mathbf{s}).
\end{equation*}
We will sometimes refer to the classical or quantum value of a Bell expression: 
\begin{equation*}
    \omega_c(\mathcal B) \coloneqq \max_{p \text{ deterministic}} 
    \sum_{\mathbf{a},\mathbf{s}} \beta_{\mathbf{a}}^{\mathbf{s}} p(\mathbf{a}|\mathbf{s})
\end{equation*}
and
\begin{equation*}
    \omega_q(\mathcal B) \coloneqq \sup_{p \text{ quantum}} 
    \sum_{\mathbf{a},\mathbf{s}} \beta_{\mathbf{a}}^{\mathbf{s}} p(\mathbf{a}|\mathbf{s})
\end{equation*}
Furthermore, given a classical or quantum strategy $\mathcal S$, we define $\omega(\mathcal B, \mathcal S):=\omega(\mathcal B, p)$ as the value of the Bell expression evaluated on the behavior $p$ induced by $\mathcal S$.
    
Fix a deterministic strategy $\mathcal S_d$, specified by
the functions $f_i:S_i\to A_i$.
Then, the \textbf{$\bm\eta$-lossy Bell expression with respect to $\mathcal S_d$} is given by:
    \begin{equation*}
    %\label{eq:lossy_Bell_operator}
    \mathcal B_{\mathcal S_d}(\bm\eta)
    \coloneqq
    \sum_{\mathbf{s},\mathbf{a}}
    \beta_{\mathbf{a}}^{\mathbf{s}}
    \sum_{T \subseteq [n]}
    \left[\prod_{i \in T} (1-\eta_i)
    \prod_{j \notin T} \eta_j\right]
    \left[\prod_{k\in T}
    \delta_{a_k,f_k(s_k)}\right]
    p
    (\mathbf{a}_{\bar T}\mid\mathbf{s}_{\bar T}),
\end{equation*}
where we used \Cref{eq:lossy_behavior}, $\bar T$ is the complement of $T$,
and $p
    (\mathbf{a}_{\bar T}\mid\mathbf{s}_{\bar T})$ is the marginal of a behavior 
\(p(\mathbf{a}|\mathbf{s})\) induced by $\mathcal S_q$, namely
\begin{equation*}
        p
    (\mathbf{a}_{\bar T}\mid\mathbf{s}_{\bar T})
    \coloneqq
    \sum_{\mathbf{a}_T}  p(\mathbf{a}|\mathbf{s}).
\end{equation*}
This marginal is well defined since we only consider
no-signaling behaviors, for which the right-hand side is
independent of $\mathbf{s}_T$. 
We note that the bipartite version of this result was presented in~\cite{gigena2024robustselftestingbellinequalities}.
We can then compute the quantum value of $\mathcal B_{\mathcal S_d}(\bm \eta)$ and thereby obtain the largest expected utility with efficiency $\bm \eta$ achievable with $\mathcal S_d$ as the fallback strategy. We will call this the \textbf{$\bm\eta$-lossy value with respect to $\mathcal S_d$} and denote it by $\omega_{\bm\eta}(\mathcal B,\mathcal S_d) \coloneqq \omega_q(\mathcal B_{\mathcal S_d}(\bm \eta))$.

Standard methods for optimizing Bell expressions can then be applied directly. In particular, upper bounds can be obtained using the Navascués--Pironio--Acín (NPA) hierarchy~\cite{navascues2008convergent}, while lower bounds can be obtained numerically using see-saw optimization~\cite{werner2001bell,Ketjl} over the shared quantum state and measurements. Finally, by repeating this optimization over all deterministic strategies $\mathcal S_d$, one obtains the lossy value:
\begin{align*}
    \omega_{\bm{\eta}}(\mathcal B) = \max_{\mathcal S_d} \omega_{\bm{\eta}}(\mathcal B,\mathcal S_d).
\end{align*}

\section{Nonlocal Games Under Loss: General Results}
\label{sec:fallback}
We now prove some facts about nonlocal games in the presence of loss, including facts about optimal fallback strategies.

For the rest of this paper, for simplicity we will
focus on the bipartite setting with \textbf{symmetric loss}, meaning both parties have the same detection efficiency: $\eta_1 =\eta_2=\eta$. 
Under symmetric loss, we call the largest achievable expected utility the $\eta$-lossy value and denote it by $\omega_\eta$. We also say the optimal fallback strategy at efficiency $\eta$ to mean the optimal fallback strategy under symmetric loss.

In this setting, we first prove a basic result regarding how the $\eta$-lossy value varies with $\eta$.
\begin{prp}
    Let $G$ be a two-player nonlocal game. 
    Assume symmetric loss and let $\omega_\eta$ be the $\eta$-lossy value. Then, $\omega_\eta$ has the following properties:
    \begin{enumerate}
        \item $\omega_0 = \omega_c$ and $\omega_1 = \omega_q$,
        \item $\omega_\eta$ is non-decreasing with $\eta$,
        \item there exists $\eta^* \in [0,1]$ such that for $\eta \leq \eta^*$, 
        \begin{align*}
            \omega_\eta = \omega_c
        \end{align*}
        and for $\eta > \eta^*$,
        \begin{align*}
            \omega_\eta > \omega_c.
        \end{align*}
    \end{enumerate}
\end{prp}
\begin{proof}
    Notice that when $\eta=0$, the lossy Bell expression reduces to an expression over deterministic strategies, and when $\eta=1$, it reduces to the usual Bell expression over quantum strategies. Therefore, $\omega_0=\omega_c, \omega_1=\omega_q$, proving the first statement.

    For $0\leq\eta<\eta'\leq1$, consider any semideterministic strategy $(\mathcal S_q, \mathcal S_d)$. Let the two functions used in $\mathcal S_d$ be $f_1,f_2$, and the quantum behavior induced by $\mathcal S_q$ be $p_q(a_1,a_2\vert s_1,s_2)$. At efficiency $\eta'$, define the quantum strategy as follows: each party independently performs the quantum measurement specified by $S_q$ with probability $\eta/\eta'$ and outputs $f_i(s_i)$ otherwise. The fallback strategy remains unchanged. The new quantum behavior can be expressed as
    \begin{align}
        p'_q(a_1,a_2\vert s_1,s_2)=&\left( \frac{\eta}{\eta'} \right)^2p_q(a_1,a_2\vert s_1,s_2)+\frac{\eta}{\eta'}\left( 1-\frac{\eta}{\eta'} \right)p_q(a_1\vert s_1)\delta_{a_2,f_2(s_2)}\notag\\
        &+\frac{\eta}{\eta'}\left( 1-\frac{\eta}{\eta'} \right)p_q(a_2\vert s_2)\delta_{a_1,f_1(s_1)}+\left( 1-\frac{\eta}{\eta'} \right)^2\delta_{a_1,f_1(s_1)}\delta_{a_2,f_2(s_2)}.
    \end{align}
    Substituting this for the quantum behavior into the definition of an $\eta'$-semideterministic behavior, we obtain the $\eta$-semideterministic behavior induced by $(\mathcal S_q, \mathcal S_d)$. Therefore, every expected utility achievable at efficiency $\eta$ is also achievable at efficiency $\eta'$. Hence, $\omega_\eta$ is non-decreasing with $\eta$, proving the second statement.
    
    Finally, recall that $\overline{\mathcal Q}$ is the closure of the set
    of finite-dimensional quantum behaviors. Since the Bell scenario is
    finite, $\overline{\mathcal Q}$ is compact. Moreover, the expected
    utility is continuous in the behavior and in $\eta$, and optimizing
    over $\mathcal Q$ or over $\overline{\mathcal Q}$ gives the same
    supremum. Since there are only finitely many deterministic fallback
    strategies, the $\eta$-lossy value can equivalently be expressed as a
    maximum of a continuous function over the compact set
    $\overline{\mathcal Q}$ and a finite set of deterministic strategies.
    Hence, by Berge's Maximum Theorem~\cite{berge1963topological},
    $\omega_\eta$ is continuous in $\eta$.    

    We now define
    $$E_c \coloneqq \{\eta \in [0,1] : \omega_{\eta} = \omega_c\}$$
    This set is closed, bounded, and non-empty since $\omega_0 =\omega_c$. Hence it has a maximum:
    \begin{align*}
        \eta^* \coloneqq \max E_c.
    \end{align*}
    For $\eta > \eta^*$, by the definition of $\eta^*$, $\omega_\eta \neq \omega_c$. Since $\omega_\eta$ is non-decreasing with $\eta$, $\omega_\eta > \omega_c$.\footnote{Note that $\eta^*$ could equal 1, which is the case when $\omega_q = \omega_c$. } For $\eta \leq \eta^*$, we know $\omega_{\eta^*} = \omega_c$. Since $\omega_\eta$ is non-decreasing with $\eta$, $\omega_\eta \leq \omega_c$. 
    However, since $\omega_0 = \omega_c$, since $\omega_\eta$ is non-decreasing, $\omega_\eta \geq \omega_c$. Hence, $\omega_\eta = \omega_c$. This proves the third statement.
\end{proof}
\noindent We call $\eta^*$ the \textbf{threshold efficiency}. 
In order to achieve a higher expected utility than $\omega_c$, we require $\eta > \eta^*$. We mention that a related but inequivalent definition of the threshold efficiency was given in~\cite{massar2003violation} where a loss event was treated as an output possibility.

Recall the $\eta$-lossy value can be written as
\begin{equation*}
%\label{eq:omega_q_star_eta}
    \omega_\eta
    =
    \max_{\mathcal S_d}
    \omega_q\!\left[
        \mathcal B_{\mathcal S_d}(\eta)
    \right],
\end{equation*}
where the maximization is over deterministic strategies. The bipartite version of the lossy Bell expression in~\cite{massar2002bell} is given by:
\begin{equation*}
        \begin{aligned}
            \mathcal B_{\mathcal S_d}(\vec\eta) \coloneqq &\eta_1\eta_2 \sum_{a_1,a_2,s_1,s_2}
\beta_{a_1,a_2}^{s_1,s_2}
p_q(a_1,a_2|s_1,s_2) + (1-\eta_1)\eta_2\sum_{a_2,s_1,s_2}\beta_{f_1(s_1),a_2}^{s_1,s_2}p_q(a_2|s_2) \\
            &+ \eta_1(1-\eta_2)\sum_{a_1,s_1,s_2}\beta_{a_1,f_2(s_2)}^{s_1,s_2}p_q(a_1|s_1) + (1-\eta_1)(1-\eta_2)\sum_{s_1,s_2}\beta_{f_1(s_1),f_2(s_2)}^{s_1,s_2}.
        \end{aligned}
        \label{eq:bi_lossyBell}
\end{equation*}
Given a semideterministic strategy $(\mathcal S_q, \mathcal S_d)$, where the functions used in $\mathcal S_d $ are $f_1, f_2$, define the \textbf{marginal term} in the lossy Bell expression by 
\begin{equation}
\label{eq:marginal}
    m_{\mathcal S_d}(\mathcal S_q) \coloneqq 
   % m_{\mathcal S_d}(p)= 
   \sum_{a_2,s_1,s_2}\beta_{f_1(s_1),a_2}^{s_1,s_2}p_q(a_2|s_2) 
            + \sum_{a_1,s_1,s_2}\beta_{a_1,f_2(s_2)}^{s_1,s_2}p_q(a_1|s_1),
\end{equation}
where $p_q$ is the behavior induced by $\mathcal S_q$. We also write $m_{\mathcal S_d}(p_q) \coloneqq m_{\mathcal S_d}(\mathcal S_q)$.

Our first observation is the following. 
\begin{prp}
    Let $\mathcal S_d$ be an optimal deterministic strategy for a Bell expression \(\mathcal B\). Then, $\omega_q\!\left(
        \mathcal B_{\mathcal S_d}(\eta)
    \right)$ is a convex function of $\eta$ for $\eta \in [0,1]$.
\end{prp}
\begin{proof}
    Given a quantum strategy $\mathcal{S}_q$, define
    \begin{equation*}
        g_{\mathcal{S}_q}(\eta) \coloneqq  \omega(\mathcal B_{\mathcal S_d}(\eta), \mathcal S_q) = \eta^2 \omega( \mathcal B,\mathcal{S}_q) + \eta(1-\eta)m_{\mathcal S_d}(\mathcal S_q) + (1-\eta)^2 \omega_c.
    \end{equation*}
    Here we used that
\begin{equation*}
        \sum_{s_1,s_2}
    \beta_{f_1(s_1),f_2(s_2)}^{s_1,s_2}
    =
    \omega_c,
\end{equation*}
because $\mathcal S_d$ is an optimal deterministic strategy for $\mathcal B$. Note that
$$
\omega_q(\mathcal B_{\mathcal S_d}(\eta)) = \sup_{\mathcal S_q} g_{\mathcal S_q}(\eta).
$$

The second derivative of $g_{\mathcal{S}_q}$ with respect to $\eta$ is 
    \begin{equation*}
        g_{\mathcal{S}_q}''(\eta) = 2\omega( \mathcal B, \mathcal{S}_q) - 2 m_{\mathcal S_d}(\mathcal S_q) + 2 \omega_c.
    \end{equation*}
    When loss happens to only one-party, by Fine's theorem~\cite{fine1982hidden}, $p_{\mathcal{S}_d} (a_1|s_1) p_{\mathcal{S}_q}(a_2|s_2)$ and $p_{\mathcal{S}_q} (a_1|s_1) p_{\mathcal{S}_d}(a_2|s_2)$ are both classical behaviors. Thus we have 
    \begin{equation*}
        m_{\mathcal S_d}(\mathcal S_q) \le 2\omega_c.
    \end{equation*}
    It follows that
    \begin{equation*}
        g_{\mathcal{S}_q}''(\eta) = 2\omega(\mathcal B, \mathcal S_q) - 2m_{\mathcal S_d}(\mathcal S_q) + 2 \omega_c \ge 2\omega(\mathcal B, \mathcal S_q) - 4\omega_c + 2 \omega_c = 2\omega(\mathcal B, \mathcal S_q) - 2\omega_c.
    \end{equation*}

Let $\eta > 0$. We claim that if 
$\omega(\mathcal B, \mathcal S_q)<\omega_c$, we cannot attain the $\eta$-lossy value. Indeed, using
$m_{\mathcal S_d}(\mathcal S_q)\le 2\omega_c$, we obtain
\begin{equation*}
\begin{aligned}
g_{\mathcal S_q}(\eta)
&\le
\eta^2
\omega(\mathcal B, \mathcal S_q)
+
2\eta(1-\eta)\omega_c
+
(1-\eta)^2\omega_c
\\
&=
\eta^2
\omega(\mathcal B, \mathcal S_q)
+
(1-\eta^2)\omega_c .
\end{aligned}
\end{equation*}
Thus, if
$\omega(\mathcal B, \mathcal S_q)<\omega_c$, then
$g_{\mathcal S_q}(\eta)<\omega_c$ for every $\eta>0$. 
However,
a trivial quantum strategy reproducing the deterministic strategy $\mathcal S_d$ achieves \(\omega_c\). Thus when we take the supremum over $\mathcal S_q$, we can also restrict to strategies for which $\omega(\mathcal B, \mathcal S_q) \geq \omega_c$.

If $\eta = 0$, $g_{\mathcal S_q}(\eta) = \omega_c$ is constant and does not depend on $\mathcal S_q$. Thus in this case we can also restrict to strategies for which $\omega(\mathcal B, \mathcal S_q) \geq \omega_c$.

Hence, we can write
\begin{equation}
\label{eq:lossy_value_given_S_d}
    \omega_q
        (\mathcal B_{\mathcal S_d}(\eta))
    =
    \sup_{\mathcal S_q : \omega(\mathcal B, \mathcal S_q) \geq \omega_c}
    g_{\mathcal S_q}(\eta),
\end{equation}
With this new domain for $\mathcal S_q$,
\begin{equation*}
    g_{\mathcal{S}_q}''(\eta) \ge 2\omega(\mathcal B, \mathcal S_q) - 2 \omega_c \ge 0.
\end{equation*}
Thus for all $\mathcal S_q$ in this domain, $g_{\mathcal S_q}(\eta)$ is convex in
$\eta$. Since a pointwise supremum of convex functions is convex~\cite[Sec.~3.2.3]{boyd2004convex}, it
follows from \Cref{eq:lossy_value_given_S_d} that
$
    \omega_q(
        \mathcal B_{\mathcal S_d}(\eta))
$
is a convex function of $\eta$.
\end{proof}

To study how the optimal deterministic fallback strategy
depends on $\eta$, we introduce an auxiliary function.
Recall that $\mathcal Q$ denotes the set of finite-dimensional quantum
behaviors and $\overline{\mathcal Q}$ its closure.
For a deterministic strategy $\mathcal S_d$, define the set of
attainable marginal contributions by
\begin{equation*}
%\label{eq:M_Sd}
    \mathcal M_{\mathcal S_d}
    \coloneqq
    \left\{
        m\in\mathbb R :
        \exists\,p \in \overline{\mathcal Q}\ \text{such that}\ 
        m=m_{\mathcal S_d}(p)
    \right\}.
\end{equation*}
Since $\mathcal Q$ is contained in the no-signaling set and the
no-signaling set is closed~\cite{brunner2014bell}, we have
$\overline{\mathcal Q}\subseteq\mathcal{NS}$. Hence
$m_{\mathcal S_d}(p)$ is well defined for every
$p\in\overline{\mathcal Q}$ via~\Cref{eq:marginal}.
Moreover, since $\mathcal Q$ is bounded and convex,
$\overline{\mathcal Q}$ is a compact and convex subset of the
finite-dimensional Euclidean space of behaviors. Since
$m_{\mathcal S_d}$ is linear in the behavior,
$\mathcal M_{\mathcal S_d}$ is a closed, finite interval in $\R$.

For each $m\in\mathcal M_{\mathcal S_d}$, define
\begin{equation*}
%\label{eq:Q_Sd}
    Q_{\mathcal S_d}(m)
    \coloneqq
    \max_{\substack{
        p\in\overline{\mathcal Q}:\\
        m_{\mathcal S_d}(p)=m
    }}
    \omega(\mathcal B,p).
\end{equation*}
Thus, $Q_{\mathcal S_d}(m)$ is the largest value of the original Bell
expression among behaviors in $\overline{\mathcal Q}$ with
marginal term equal to $m$.

\begin{prp}
For any deterministic strategy $\mathcal S_d$, the function
$Q_{\mathcal S_d}$ is concave on $\mathcal M_{\mathcal S_d}$.
\end{prp}
\begin{proof}
Let $m_1,m_2\in\mathcal M_{\mathcal S_d}$ and let $k\in[0,1]$.
Choose behaviors
$p_1,p_2 \in \overline{\mathcal Q}$ such that
\begin{equation*}
    m_{\mathcal S_d}(p_i)=m_i
\end{equation*}
and 
\begin{equation*}
    \omega(\mathcal B,p_i)
    =
    Q_{\mathcal S_d}(m_i),
    \qquad i=1,2.
\end{equation*}

Let
$$p(k) \coloneqq k p_1 + (1-k) p_2.$$
Since $\overline{\mathcal Q}$ is convex, $p(k) \in \overline{\mathcal Q}$. As
$\omega(\mathcal B,p)$ is linear in the behavior,
\begin{equation*}
    m_{\mathcal S_d}(p(k))
    =
    km_1+(1-k)m_2
\end{equation*}
and
\begin{equation*}
\begin{aligned}
    \omega(\mathcal B,p(k))
    &=
    k\omega(\mathcal B,p_1)
    +(1-k)\omega(\mathcal B, p_2)
    \\
    &=
    kQ_{\mathcal S_d}(m_1)
    +(1-k)Q_{\mathcal S_d}(m_2)
    .
\end{aligned}
\end{equation*}
By the definition of $Q_{\mathcal S_d}$,
\begin{equation*}
    Q_{\mathcal S_d}\!\left(km_1+(1-k)m_2\right)
    \ge \omega(\mathcal B,p(k))
    %kQ_{\mathcal S_d}(m_1)
    %+(1-k)Q_{\mathcal S_d}(m_2)
   .
\end{equation*}
Thus,
\begin{equation*}
    Q_{\mathcal S_d}\!\left(km_1+(1-k)m_2\right)
    \ge
    kQ_{\mathcal S_d}(m_1)
    +(1-k)Q_{\mathcal S_d}(m_2).
\end{equation*}
Hence $Q_{\mathcal S_d}$ is concave on
$\mathcal M_{\mathcal S_d}$.

\end{proof}

This concavity property allows us to characterize the marginal
terms that optimize the $\eta$-lossy value for a fixed deterministic strategy $\mathcal S_d$.
Using $Q_{\mathcal S_d}$, we can write
\begin{equation}
\label{eq:omega_q_eta_Q}
\omega_q\!\left(
    \mathcal B_{\mathcal S_d}(\eta)
\right)
=
\max_{m\in\mathcal M_{\mathcal S_d}}
    \eta^2 Q_{\mathcal S_d}(m)
    +
    \eta(1-\eta)m
    +
    (1-\eta)^2\omega(\mathcal B,\mathcal S_d)
.
\end{equation}
For $\eta>0$, the last term in
\Cref{eq:omega_q_eta_Q} is independent of $m$, and multiplication by
the positive factor $\eta^2$ does not affect the optimization.
Thus it suffices to consider
\begin{equation}
\label{eq:aux_m_optimization}
\max_{m\in\mathcal M_{\mathcal S_d}}
    Q_{\mathcal S_d}(m)
    +
    \frac{1-\eta}{\eta}m .
\end{equation}
If $\mathcal M_{\mathcal S_d}$ is a singleton, the optimization in
\Cref{eq:aux_m_optimization} is trivial. Otherwise, 
let $m_-, m_+ \in \R$ such that
\[
\mathcal M_{\mathcal S_d}=[m_-,m_+].
\]
Since $Q_{\mathcal S_d}$ is concave, its left and right derivatives, 
$Q'_{\mathcal S_d,-}(m)$ and
$Q'_{\mathcal S_d,+}(m)$ respectively, exist at every $m\in(m_-,m_+)$ 
of
$\mathcal M_{\mathcal S_d}$ and $Q'_{\mathcal S_d,-}(m) \geq Q'_{\mathcal S_d,+}(m)$.
At the endpoints, the corresponding one-sided derivatives
$Q'_{\mathcal S_d,+}(m_-)$ and
$Q'_{\mathcal S_d,-}(m_+)$ exist as well~\cite[Thm.~23.1]{rockafellar1970convex}.

It then follows from the concavity of
$Q_{\mathcal S_d}$ that a point
$m\in\mathcal M_{\mathcal S_d}$ maximizes \Cref{eq:aux_m_optimization} if and only if
\begin{equation*}
\begin{cases}
Q'_{\mathcal S_d,+}(m_-)
    \leq -\dfrac{1-\eta}{\eta},
    & m=m_-, \\[0.8em]
-\dfrac{1-\eta}{\eta}
    \in
    \left[
        Q'_{\mathcal S_d,+}(m),
        Q'_{\mathcal S_d,-}(m)
    \right],
    & m\in(m_-,m_+), \\[0.8em]
Q'_{\mathcal S_d,-}(m_+)
    \geq -\dfrac{1-\eta}{\eta},
    & m=m_+.
\end{cases}
\end{equation*}
We call any such $m$ an \textbf{optimal marginal value} at efficiency
$\eta$, and denote the set of all optimal marginal values by
\begin{equation*}
%\label{eq:optimal_marginal_set}
    \mathcal M_{\mathcal S_d}^{\mathrm{opt}}(\eta)
    \coloneqq
    \operatorname*{arg\,max}_{m\in\mathcal M_{\mathcal S_d}}
    \left\{
        Q_{\mathcal S_d}(m)
        +
        \frac{1-\eta}{\eta}m
    \right\}.
\end{equation*}

\begin{lem}
\label[lem]{lem:opt_marg}
    Let $\mathcal S_d$ be a deterministic strategy.
For any $0<\eta_1<\eta_2\le 1$, if
\[
    m_1\in\mathcal M_{\mathcal S_d}^{\mathrm{opt}}(\eta_1),
    \qquad
    m_2\in\mathcal M_{\mathcal S_d}^{\mathrm{opt}}(\eta_2),
\]
then
\[
    m_1\ge m_2.
\]
In other words, the optimal marginal values are non-increasing with
the efficiency $\eta$.
\end{lem}

\begin{proof}
Let $0<\eta_1<\eta_2\le 1$, and choose
\begin{equation*}
    m_i\in\mathcal M_{\mathcal S_d}^{\mathrm{opt}}(\eta_i),
    \qquad i=1,2.
\end{equation*}
By the optimality of $m_1$ at efficiency $\eta_1$,
\begin{equation*}
    Q_{\mathcal S_d}(m_1)
    +
    \frac{1-\eta_1}{\eta_1}m_1
    \ge
    Q_{\mathcal S_d}(m_2)
    +
    \frac{1-\eta_1}{\eta_1}m_2.
\end{equation*}
Similarly, the optimality of $m_2$ at efficiency $\eta_2$ gives
\begin{equation*}
    Q_{\mathcal S_d}(m_2)
    +
    \frac{1-\eta_2}{\eta_2}m_2
    \ge
    Q_{\mathcal S_d}(m_1)
    +
    \frac{1-\eta_2}{\eta_2}m_1.
\end{equation*}
Adding the two inequalities yields
\begin{equation*}
    \left(
        \frac{1-\eta_1}{\eta_1}
        -
        \frac{1-\eta_2}{\eta_2}
    \right)
    (m_1-m_2)
    \ge 0.
\end{equation*}
Since the function
\[
    \eta\longmapsto\frac{1-\eta}{\eta}
\]
is strictly decreasing on $(0,1]$, we have
\begin{equation*}
    \frac{1-\eta_1}{\eta_1}
    -
    \frac{1-\eta_2}{\eta_2}
    >0.
\end{equation*}
Therefore,
\begin{equation*}
    m_1\ge m_2,
\end{equation*}
which proves the claim.
\end{proof}

For each deterministic strategy $\mathcal S_d$, define
\begin{equation}
\label{eq:m_Sd}
    m_{\mathcal S_d}^{q}
    \coloneqq
    \max
    \left\{
        m\in\mathcal M_{\mathcal S_d}:
        Q_{\mathcal S_d}(m)=\omega_q(\mathcal B)
    \right\}.
\end{equation}
The set in \Cref{eq:m_Sd} is nonempty and compact. Indeed,
$\overline{\mathcal Q}$ is compact, and the continuous Bell functional
$\mathcal B$ attains its maximum $\omega_q(\mathcal B)$ on
$\overline{\mathcal Q}$. Thus, $m_{\mathcal S_d}^{q}$ is the largest
marginal contribution associated with $\mathcal S_d$ among all optimal
quantum behaviors $p$ in $\overline{\mathcal Q}$ satisfying $\omega(\mathcal B,p) = \omega_q(\mathcal{B})$. We now prove a robustness theorem that shows an optimal fallback strategy remains optimal in a neighborhood of $\eta=1$.

\begin{thm}
\label[thm]{thm:optimal_fallback_near_one}
Let $\mathcal B$ be a Bell expression. Suppose that there exists a
unique deterministic strategy $\mathcal S_d^0$ maximizing
$m_{\mathcal S_d}^{q}$, i.e.,
\begin{equation*}
    m_{\mathcal S_d^0}^{q}
    >
    m_{\mathcal S_d}^{q}
    \qquad
    \text{for every }
    \mathcal S_d\neq\mathcal S_d^0.
\end{equation*}
Then there exists $\varepsilon>0$ such that
$\mathcal S_d^0$ is the unique optimal fallback strategy
for every $\eta\in(1-\varepsilon,1)$.
\end{thm}

\begin{proof}
For any deterministic strategy $\mathcal S_d$,
\begin{equation*}
\begin{aligned}
\omega_q\!\left(
    \mathcal B_{\mathcal S_d}(\eta)
\right)
=
\max_{p\in\overline{\mathcal Q}}
\eta^2\omega(\mathcal B,p)
+\eta(1-\eta)m_{\mathcal S_d}(p)
+(1-\eta)^2\omega(\mathcal B,\mathcal S_d).
\end{aligned}
\end{equation*}
For fixed $p$, the derivative of the objective function
with respect to $\eta$ at $\eta=1$ is
\begin{equation*}
    2\omega(\mathcal B,p)-m_{\mathcal S_d}(p).
\end{equation*}

Since $\overline{\mathcal Q}$ is compact, and the objective function is
continuous in $(\eta,p)$ and continuously differentiable in $\eta$,
Danskin's theorem~\cite[Thm.~3.1]{danskin1967theory} applies. At $\eta=1$, the
maximizing behaviors are precisely those satisfying
$\omega(\mathcal B,p)=\omega_q$. Hence the left derivative at $\eta=1$ is
\begin{equation*}
\begin{aligned}
\left.
\frac{d}{d\eta}
\omega_q\!\left(
    \mathcal B_{\mathcal S_d}(\eta)
\right)
\right|_{\eta=1^-}
&=
\min_{\substack{
    p\in\overline{\mathcal Q}:\\
    \omega(\mathcal B,p)=\omega_q(\mathcal B)
}}
\left\{
    2\omega(\mathcal B,p)-m_{\mathcal S_d}(p)
\right\}
\\
&=
2\omega_q(\mathcal B)-m_{\mathcal S_d}^{q},
\end{aligned}
\end{equation*}
where the last equality follows from \Cref{eq:m_Sd}.
Consequently, as $\eta\to1^-$,
\begin{equation}
\label{eq:lossy_expansion_near_one}
\omega_q\!\left(
    \mathcal B_{\mathcal S_d}(\eta)
\right)
=
\omega_q(\mathcal B)
+
(1-\eta)
\left(
    m_{\mathcal S_d}^{q}
    -2\omega_q(\mathcal B)
\right)
+
r(\eta).
\end{equation}
Here, $r(\eta)$ is a remainder term that satisfies
\begin{equation}
\label{eq:remainder}    
\lim_{\eta \to 1^-} \frac{r(\eta)}{1-\eta} = 0.
\end{equation}

Now fix any deterministic strategy
$\mathcal S_d\neq\mathcal S_d^0$. Applying
\Cref{eq:lossy_expansion_near_one} to
$\mathcal S_d^0$ and $\mathcal S_d$ gives
\begin{equation*}
\begin{aligned}
&
\omega_q\!\left(
    \mathcal B_{\mathcal S_d^0}(\eta)
\right)
-
\omega_q\!\left(
    \mathcal B_{\mathcal S_d}(\eta)
\right)
\\
&\qquad =
(1-\eta)
\left(
    m_{\mathcal S_d^0}^{q}
    -
    m_{\mathcal S_d}^{q}
\right)
+
r'(\eta),
\end{aligned}
\end{equation*}
where $r'(\eta)$ also satisfies~\Cref{eq:remainder}.
By assumption,
\begin{equation*}
    m_{\mathcal S_d^0}^{q}
    >
    m_{\mathcal S_d}^{q}.
\end{equation*}
Hence there exists $\varepsilon_{\mathcal S_d}>0$ such that
\begin{equation*}
\omega_q\!\left(
    \mathcal B_{\mathcal S_d^0}(\eta)
\right)
>
\omega_q\!\left(
    \mathcal B_{\mathcal S_d}(\eta)
\right)
\end{equation*}
for every
$\eta\in(1-\varepsilon_{\mathcal S_d},1)$.

Since there are only finitely many deterministic strategies, we may
choose
\begin{equation*}
    \varepsilon
    =
    \min_{\mathcal S_d\neq\mathcal S_d^0}
    \varepsilon_{\mathcal S_d}
    >0.
\end{equation*}
Therefore $\mathcal S_d^0$ is the unique optimal 
fallback strategy for every $\eta\in(1-\varepsilon,1)$.
\end{proof}

The same argument, without the uniqueness assumption in
\Cref{thm:optimal_fallback_near_one}, yields the following corollary.
\begin{cor}
\label[cor]{cor:optimal_fallback_near_one}
Let $\mathcal B$ be a Bell expression, and let
\begin{equation*}
    \mathfrak S^{q}
    :=
    \left\{
        \mathrm {fallback\,strategy\,} \mathcal S_d :
        m_{\mathcal S_d}^{q}
        =
        \max_{\mathcal S_d'} m_{\mathcal S_d'}^{q}
    \right\}.
\end{equation*}
Then there exists $\varepsilon>0$ such that, for every
$\eta\in(1-\varepsilon,1)$, any optimal fallback strategy belongs to
$\mathfrak S^{q}$. Hence, for $\eta$ sufficiently close to $1$, it
suffices to consider the fallback strategies in $\mathfrak S^{q}$
when determining the optimal fallback strategy.
\end{cor}

We derive the upper and lower bounds on the threshold efficiency
$\eta^*$ separately.

\begin{prp}
\label[prp]{prp:threshold_efficiency_upper_bound}
Let $\mathcal B$ be a Bell expression for which
$\omega_q>\omega_c$. Choose
\begin{equation*}
%\label{eq:Sd_opt}
    \mathcal S_d^{\mathrm{opt}}
    \in
    \operatorname*{arg\,max}_{\substack{
        \mathcal S_d:\\
        \omega(\mathcal B,\mathcal S_d)=\omega_c
    }}
    m_{\mathcal S_d}^{q},
\end{equation*}
and define
\begin{equation*}
    m_{\mathrm{opt}}
    \coloneqq
    m_{\mathcal S_d^{\mathrm{opt}}}^{q}.
\end{equation*}
Then the threshold efficiency satisfies
\begin{equation}
\label{eq:threshold_efficiency_upper_bound}
    \eta^*
    \le
    \frac{
        2\omega_c-m_{\mathrm{opt}}
    }{
        \omega_q+\omega_c-m_{\mathrm{opt}}
    }.
\end{equation}
\end{prp}

\begin{proof}
By the definition of $m_{\mathrm{opt}}$, there exists
$p^{*}\in\overline{\mathcal Q}$ such that
\begin{equation*}
    \omega(\mathcal B,p^{*})
    =
    \omega_q,
    \qquad
    m_{\mathcal S_d^{\mathrm{opt}}}(p^{*})
    =
    m_{\mathrm{opt}}.
\end{equation*}
Moreover, by the definition of
$\mathcal S_d^{\mathrm{opt}}$,
\begin{equation*}
    \omega(\mathcal B,\mathcal S_d^{\mathrm{opt}})
    =
    \omega_c.
\end{equation*}
Therefore,
\begin{equation*}
\begin{aligned}
\omega_q\!\left(
    \mathcal B_{\mathcal S_d^{\mathrm{opt}}}(\eta)
\right)
&\ge
\eta^2\omega_q
+
\eta(1-\eta)m_{\mathrm{opt}}
+
(1-\eta)^2\omega_c.
\end{aligned}
\end{equation*}
Subtracting $\omega_c$ from the right-hand side gives
\begin{equation*}
    \eta
    \left[
        \eta\bigl(
            \omega_q+\omega_c-m_{\mathrm{opt}}
        \bigr)
        -
        \bigl(
            2\omega_c-m_{\mathrm{opt}}
        \bigr)
    \right].
\end{equation*}
Using $m_{\mathrm{opt}}\le 2\omega_c$ and
$\omega_q>\omega_c$, we have
\begin{equation*}
    \omega_q+\omega_c-m_{\mathrm{opt}}
    \ge
    \omega_q-\omega_c
    >0.
\end{equation*}
Hence
\begin{equation*}
\omega_q\!\left(
    \mathcal B_{\mathcal S_d^{\mathrm{opt}}}(\eta)
\right)
>
\omega_c
\end{equation*}
whenever
\begin{equation*}
    \eta
    >
    \frac{
        2\omega_c-m_{\mathrm{opt}}
    }{
        \omega_q+\omega_c-m_{\mathrm{opt}}
    }.
\end{equation*}
Thus quantum advantage is achieved for every efficiency above the
right-hand side, which proves
\Cref{eq:threshold_efficiency_upper_bound}.
\end{proof}

%%%%%%%%%%%%%
\begin{prp}
\label[prp]{prp:threshold_efficiency_lower_bound}
Let $\mathcal B$ be a Bell expression for which
$\omega_q>\omega_c$. For each deterministic strategy
$\mathcal S_d$, suppose that the set
\begin{equation*}
    \left\{
        m\in\mathcal M_{\mathcal S_d}:
        m>m_{\mathcal S_d}^{q},
        \quad
        Q_{\mathcal S_d}(m)=\omega_c
    \right\}
\end{equation*}
has a minimum, and define
\begin{equation}
\label{eq:m_Sd_c}
    m_{\mathcal S_d}^{c}
    \coloneqq
    \min
    \left\{
        m\in\mathcal M_{\mathcal S_d}:
        m>m_{\mathcal S_d}^{q},
        \quad
        Q_{\mathcal S_d}(m)=\omega_c
    \right\}.
\end{equation}
Suppose further that
$Q'_{\mathcal S_d,-}(m_{\mathcal S_d}^{c})$
is finite for every $\mathcal S_d$. Then
\begin{equation}
\label{eq:threshold_efficiency_lower_bound}
    \min_{\mathcal S_d}
    \frac{1}{
        1-Q'_{\mathcal S_d,-}(m_{\mathcal S_d}^{c})
    }
    \le
    \eta^*.
\end{equation}
\end{prp}

\begin{proof}
Fix a deterministic strategy $\mathcal S_d$. For brevity, write
\begin{equation*}
    Q(m)
    \coloneqq
    Q_{\mathcal S_d}(m),
    \qquad
    m_q
    \coloneqq
    m_{\mathcal S_d}^{q},
    \qquad
    m_c
    \coloneqq
    m_{\mathcal S_d}^{c},
\end{equation*}
and denote
\begin{equation*}
    q_c
    \coloneqq
    Q'_-(m_c).
\end{equation*}
By definition,
\begin{equation*}
    m_c>m_q,
    \qquad
    Q(m_c)=\omega_c,
    \qquad
    Q(m_q)=\omega_q.
\end{equation*}
Since $Q$ is concave and $m_c>m_q$, we have
\begin{equation*}
\begin{aligned}
    q_c
    &\le
    \frac{
        Q(m_c)-Q(m_q)
    }{
        m_c-m_q
    }
    \\
    &=
    \frac{
        \omega_c-\omega_q
    }{
        m_c-m_q
    }
    <0.
\end{aligned}
\end{equation*}
In particular,
\begin{equation*}
    1-q_c>1.
\end{equation*}

By concavity, the left derivative $q_c$ defines a supporting line of
$Q$ at $m_c$, so that
\begin{equation}
\label{eq:Q_support_mc}
    Q(m)
    \le
    \omega_c
    +
    q_c(m-m_c)
\end{equation}
for every $m\in\mathcal M_{\mathcal S_d}$.

Suppose that
\begin{equation*}
    0<\eta
    \le
    \frac{1}{1-q_c}.
\end{equation*}
Equivalently,
\begin{equation}
\label{eq:slope_condition_fixed_Sd}
    q_c+\frac{1-\eta}{\eta}\ge0.
\end{equation}

We show that under this condition,
\begin{equation*}
    \omega_q\!\left(
        \mathcal B_{\mathcal S_d}(\eta)
    \right)
    \le
    \omega_c.
\end{equation*}

First, consider $m<m_c$. Using
\Cref{eq:Q_support_mc}, we obtain
\begin{align*}
&
\eta^2 Q(m)
+
\eta(1-\eta)m
\\
&\le
\eta^2\omega_c
+
\eta(1-\eta)m_c
+
\eta^2
\left(
    q_c+\frac{1-\eta}{\eta}
\right)
(m-m_c)
\\
&\le
\eta^2\omega_c
+
\eta(1-\eta)m_c,
\end{align*}
where the last inequality follows from
$m-m_c<0$ and
\Cref{eq:slope_condition_fixed_Sd}.
Therefore,
\begin{align*}
&
\eta^2 Q(m)
+
\eta(1-\eta)m
+
(1-\eta)^2
\omega(\mathcal B,\mathcal S_d)
\\
&\le
\eta^2\omega_c
+
\eta(1-\eta)m_c
+
(1-\eta)^2\omega_c
\\
&\le
\omega_c,
\end{align*}
where we used
$m_c\le2\omega_c$ and
$\omega(\mathcal B,\mathcal S_d)\le\omega_c$.

Next, consider $m\ge m_c$. Since $q_c<0$,
\Cref{eq:Q_support_mc} implies
\begin{equation*}
    Q(m)
    \le
    \omega_c.
\end{equation*}
Using also
$m\le2\omega_c$ and
$\omega(\mathcal B,\mathcal S_d)\le\omega_c$, we obtain
\begin{align*}
&
\eta^2 Q(m)
+
\eta(1-\eta)m
+
(1-\eta)^2
\omega(\mathcal B,\mathcal S_d)
\\
&\le
\eta^2\omega_c
+
2\eta(1-\eta)\omega_c
+
(1-\eta)^2\omega_c
\\
&=
\omega_c.
\end{align*}

Thus, for every
$m\in\mathcal M_{\mathcal S_d}$,
\begin{equation*}
    \eta^2 Q(m)
    +
    \eta(1-\eta)m
    +
    (1-\eta)^2
    \omega(\mathcal B,\mathcal S_d)
    \le
    \omega_c.
\end{equation*}
Taking the maximum over
$m\in\mathcal M_{\mathcal S_d}$ gives
\begin{equation*}
    \omega_q\!\left(
        \mathcal B_{\mathcal S_d}(\eta)
    \right)
    \le
    \omega_c
\end{equation*}
whenever
\begin{equation*}
    0<\eta
    \le
    \frac{1}{1-q_c}.
\end{equation*}
The same inequality trivially holds at $\eta=0$.

Since this argument applies to every deterministic strategy
$\mathcal S_d$, no fallback strategy can yield a lossy quantum value
larger than $\omega_c$ whenever
\begin{equation*}
    \eta
    \le
    \min_{\mathcal S_d}
    \frac{1}{
        1-Q'_{\mathcal S_d,-}(m_{\mathcal S_d}^{c})
    }.
\end{equation*}
Therefore,
\begin{equation*}
    \min_{\mathcal S_d}
    \frac{1}{
        1-Q'_{\mathcal S_d,-}(m_{\mathcal S_d}^{c})
    }
    \le
    \eta^*,
\end{equation*}
which proves the claim.
\end{proof}

We next show that the derivative appearing in
\Cref{prp:threshold_efficiency_lower_bound} can be bounded using a
linear program.

Fix a deterministic strategy $\mathcal S_d$, and use the shorthand
\begin{equation*}
    Q(m)\coloneqq Q_{\mathcal S_d}(m),
    \qquad
    m(p)\coloneqq m_{\mathcal S_d}(p),
    \qquad
    m_c\coloneqq m_{\mathcal S_d}^{c}.
\end{equation*}
Since $Q$ is concave and $Q(m_c)=\omega_c$, its left derivative can be
written as
\begin{equation}
\label{eq:left_derivative_secant}
    Q'_-(m_c)
    =
    -
    \sup_{m<m_c}
    \frac{
        Q(m)-\omega_c
    }{
        m_c-m
    }.
\end{equation}

Let $\mathcal{NS}$ denote the set of no-signaling behaviors, and define
\begin{equation*}
    Q^{\mathrm{NS}}(m)
    \coloneqq
    \max_{\substack{
        p\in\mathcal{NS}:\\
        m(p)=m
    }}
    \omega(\mathcal B,p).
\end{equation*}
Since $\overline{\mathcal Q}\subseteq\mathcal{NS}$,
\begin{equation*}
    Q(m)
    \le
    Q^{\mathrm{NS}}(m)
\end{equation*}
for every $m$. Therefore,
\begin{equation*}
    Q'_-(m_c)
    \ge
    -
    \sup_{m<m_c}
    \frac{
        Q^{\mathrm{NS}}(m)-\omega_c
    }{
        m_c-m
    }.
\end{equation*}
Equivalently,
\begin{equation}
\label{eq:NS_fractional_program}
    Q'_-(m_c)
    \ge
    -
    \sup_{\substack{
        p\in\mathcal{NS}:\\
        m(p)<m_c
    }}
    \frac{
        \omega(\mathcal B,p)-\omega_c
    }{
        m_c-m(p)
    }.
\end{equation}

Since both $\omega(\mathcal B,p)$ and $m(p)$ are linear in $p$, while
the no-signaling set is defined by linear constraints, the
linear-fractional optimization in
\Cref{eq:NS_fractional_program} can be converted into a linear program
using the Charnes--Cooper transformation~\cite{charnes1962programming}.
Introduce
\begin{equation*}
    t
    \coloneqq
    \frac{1}{
        m_c-m(p)
    },
    \qquad
    \widetilde p
    \coloneqq
    t p .
\end{equation*}
Then
\begin{equation*}
    m_c t-m(\widetilde p)=1,
\end{equation*}
and
\begin{equation*}
    \frac{
        \omega(\mathcal B,p)-\omega_c
    }{
        m_c-m(p)
    }
    =
    \omega(\mathcal B,\widetilde p)-\omega_c t.
\end{equation*}

Define $\ell_{\mathrm{NS}}(\mathcal S_d)$ as the optimal value of the
following linear program:
\begin{equation}
\label{eq:NS_derivative_LP}
\begin{aligned}
\ell_{\mathrm{NS}}(\mathcal S_d)
=
\max_{\widetilde p,t}\quad
&
    \omega(\mathcal B,\widetilde p)-\omega_c t
\\
\text{subject to}\quad
&
    m_c t-m(\widetilde p)=1,
\\
&
    \sum_{\mathbf a}
    \widetilde p(\mathbf a|\mathbf s)
    =
    t,
    \qquad
    \forall\,\mathbf s,
\\
&
    \widetilde p
    \text{ satisfies the no-signaling constraints},
\\
&
    \widetilde p(\mathbf a|\mathbf s)\geq0,
    \qquad
    \forall\,\mathbf a,\mathbf s,
\\
&
    t\geq0.
\end{aligned}
\end{equation}
The normalization constraints are scaled by $t$, while the
no-signaling constraints remain homogeneous under the transformation.
The optimal value therefore gives
\begin{equation}
\label{eq:NS_derivative_bound}
    Q'_{\mathcal S_d,-}(m_{\mathcal S_d}^{c})
    \geq
    -\ell_{\mathrm{NS}}(\mathcal S_d).
\end{equation}

Combining \Cref{prp:threshold_efficiency_lower_bound} with
\Cref{eq:NS_derivative_bound} gives the computable lower bound
\begin{equation}
\label{eq:NS_threshold_lower_bound}
    \eta^*
    \geq
    \min_{\mathcal S_d}
    \frac{1}{
        1+\ell_{\mathrm{NS}}(\mathcal S_d)
    }.
\end{equation}

As an example, consider the CHSH game. The inputs
$x,y\in\{0,1\}$ are uniformly distributed, and the players win if
their outputs $a,b\in\{0,1\}$ satisfy
\begin{equation*}
    a\oplus b=xy.
\end{equation*}
Its classical and quantum values are
\begin{equation*}
    \omega_c=\frac{3}{4},
    \qquad
    \omega_q=\frac{2+\sqrt{2}}{4}.
\end{equation*}

Choose the fallback strategy
\begin{equation*}
    f_A(x)=f_B(y)=0,
    \qquad
    x,y\in\{0,1\}.
\end{equation*}
For this strategy,
\begin{equation*}
    m_{\mathcal S_d}(p)
    =
    \frac{1}{2}
    +
    \frac{1}{2}p_A(0|0)
    +
    \frac{1}{2}p_B(0|0),
\end{equation*}
and hence
\begin{equation*}
    m_{\mathcal S_d}(p)
    \leq
    \frac{3}{2}.
\end{equation*}

At $m=3/2$, the no-signaling constraints imply that the CHSH value is
at most $\omega_c=3/4$, and this value is attained by the deterministic
behavior $a=b=0$. Thus,
\begin{equation*}
    Q_{\mathcal S_d}\left(\frac{3}{2}\right)
    =
    \omega_c.
\end{equation*}
Moreover, for an optimal CHSH quantum behavior,
\begin{equation*}
    m_{\mathcal S_d}^{q}=1,
    \qquad
    Q_{\mathcal S_d}(m_{\mathcal S_d}^{q})
    =
    \omega_q
    >
    \omega_c.
\end{equation*}
By the concavity of $Q_{\mathcal S_d}$, it follows that
\begin{equation*}
    m_{\mathcal S_d}^{c}
    =
    \frac{3}{2}.
\end{equation*}

For this fallback strategy, the linear program
\Cref{eq:NS_derivative_LP} becomes
\begin{equation*}
\begin{aligned}
\ell_{\mathrm{NS}}(\mathcal S_d)
=
\max_{\widetilde p,t}\quad
&
    \frac{1}{4}
    \sum_{\substack{
        a,b,x,y:\\
        a\oplus b=xy
    }}
    \widetilde p(a,b|x,y)
    -
    \frac{3}{4}t
\\
\text{subject to}\quad
&
    t
    -
    \frac{1}{2}\sum_b \widetilde p(0,b|0,0)
    -
    \frac{1}{2}\sum_a \widetilde p(a,0|0,0)
    =
    1,
\\
&
    \sum_{a,b}\widetilde p(a,b|x,y)=t,
    \qquad
    \forall\,x,y,
\\
&
    \widetilde p
    \text{ satisfies the no-signaling constraints},
\\
&
    \widetilde p(a,b|x,y)\geq0,
    \qquad
    \forall\,a,b,x,y,
\\
&
    t\geq0.
\end{aligned}
\end{equation*}

The no-signaling polytope in the CHSH scenario has $24$ vertices:
$16$ local deterministic behaviors and $8$ PR-box behaviors
~\cite{popescu1994quantum,brunner2014bell}.
Since the linear-fractional objective in
\Cref{eq:NS_fractional_program} attains its optimum at a vertex of the
no-signaling polytope, evaluating it on these $24$ vertices gives
\begin{equation*}
    \ell_{\mathrm{NS}}(\mathcal S_d)
    =
    \frac{1}{2},
\end{equation*}
and hence
\begin{equation*}
    Q'_{\mathcal S_d,-}
    \left(m_{\mathcal S_d}^{c}\right)
    \geq
    -\frac{1}{2}.
\end{equation*}
Therefore,
\begin{equation*}
    \eta^*
    \geq
    \frac{1}{
        1+\ell_{\mathrm{NS}}(\mathcal S_d)
    }
    =
    \frac{2}{3}.
\end{equation*}
Similarly, the no-signaling linear programs for the other deterministic
fallback strategies give the same lower bound,
$\eta^*\ge2/3$.

For the upper bound, \Cref{prp:threshold_efficiency_upper_bound} gives
\begin{equation*}
    \eta^*
    \leq
    \frac{
        2\omega_c-m_{\mathrm{opt}}
    }{
        \omega_q+\omega_c-m_{\mathrm{opt}}
    }
    =
    2\sqrt{2}-2.
\end{equation*}
The exact threshold efficiency for CHSH is known to be
$\eta^*=2/3$~\cite{PhysRevA.47.R747}, showing that the no-signaling
lower bound obtained above is tight.
%%%%%%%%%%%%%

\section{Bell Inequality Case Studies} 
\label{sec:case_studies}
In this section, we focus on two main questions concerning Bell nonlocality under loss. 
First, we investigate how the optimal fallback strategy $\mathcal S_d^*$ depends
on the efficiency $\eta$. This is interesting because while we can tune $\eta \in [0,1]$ continuously, $\mathcal S_d^*$ is a discrete object. 
For special cases, such as the CHSH game, this
dependence can be analyzed analytically. For other Bell inequalities,
however, a general criterion for identifying $\mathcal S_d^*$ is
not known. We therefore study representative examples of Bell inequalities numerically by computing  the lossy value with respect to different fallback strategies. Second, we also estimate the threshold efficiency $\eta^*$,
defined as the efficiency below which the lossy value does not
exceed the classical value.

\subsection{The CHSH game}
In the CHSH game, two parties, Alice and Bob, receive inputs
$x,y\in\{0,1\}$ uniformly at random and output bits $a,b\in\{0,1\}$. They win if
\begin{equation*}
    a\oplus b = xy.
\end{equation*}
Numerically, it is observed that optimal deterministic strategies outperform non-optimal deterministic strategies when used as the fallback strategy for the CHSH game~\cite{gigena2024robustselftestingbellinequalities}, as shown in \Cref{fig:chsh_compare}. We rigorously prove this result.
\begin{figure}
    \centering
    \includegraphics[width=0.5\linewidth]{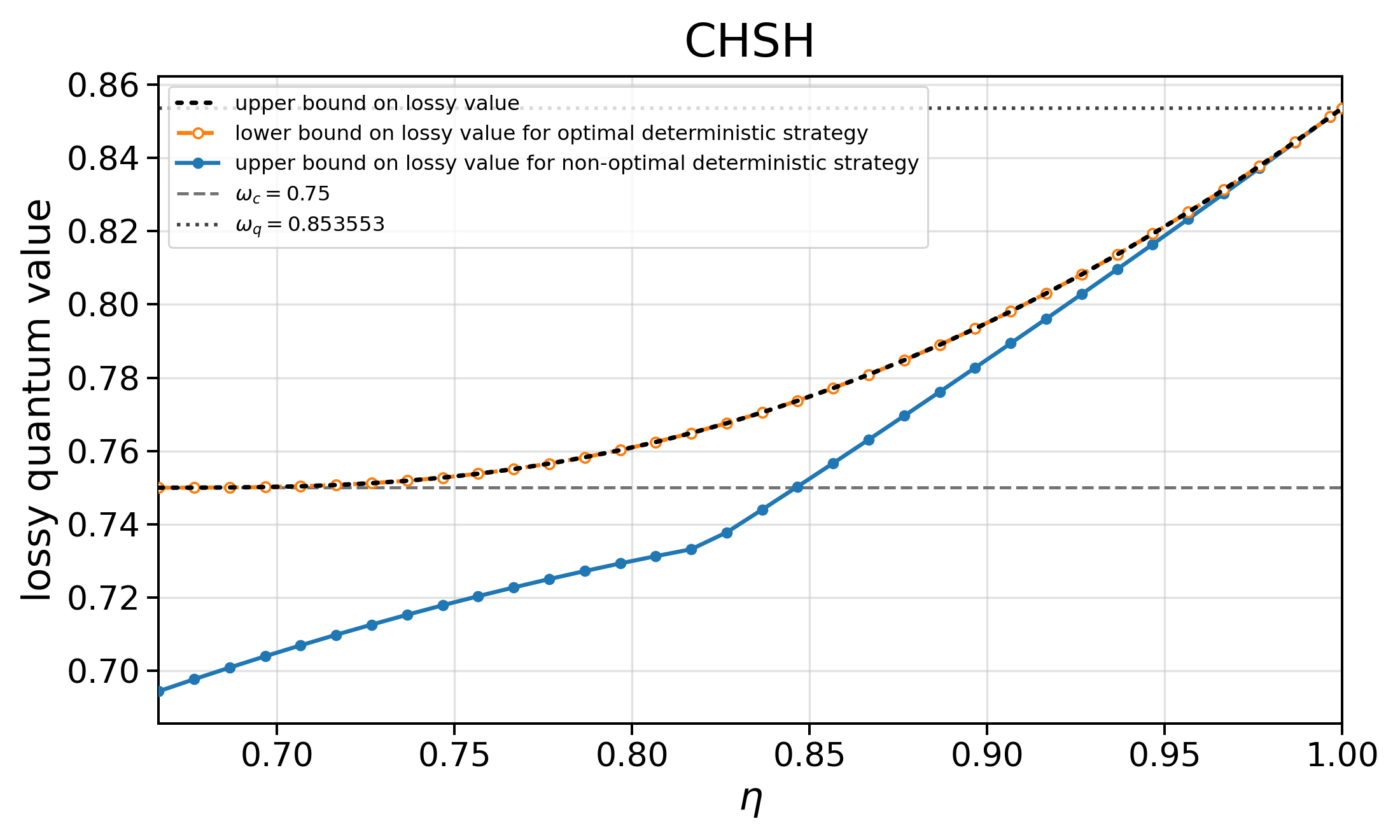}
    \caption{Comparison between optimal and non-optimal deterministic strategies as fallback strategies for the CHSH game under loss. For the non-optimal deterministic strategy, an upper bound on the $\eta$-lossy value is obtained using the NPA hierarchy, whereas for
the optimal deterministic strategy, a lower bound is obtained using
the see-saw method.}
    \label{fig:chsh_compare}
\end{figure}

We first rewrite the winning probability of a quantum strategy for the CHSH game using correlators. 
Let
\begin{equation*}
\begin{aligned}
    &\langle A_xB_y\rangle
    \coloneqq 
    p(00|xy)-p(01|xy)-p(10|xy)+p(11|xy),
    \\
    &\langle A_x\rangle
    \coloneqq
    p_A(0|x)-p_A(1|x),
    \qquad
    \langle B_y\rangle
    \coloneqq 
    p_B(0|y)-p_B(1|y),
\end{aligned}
\end{equation*}
where
\begin{equation*}
    p_A(a|x)
    =
    \sum_b p(a,b|x,y),
    \qquad
    p_B(b|y)
    =
    \sum_a p(a,b|x,y).
\end{equation*}
The marginals are well defined as we consider no-signaling
behaviors.
With this convention, the CHSH winning probability can be
written using correlators:
\begin{equation*}
    p_{\mathrm{win}}
    =
    \frac12+
    \frac18
    \sum_{x,y\in\{0,1\}}
    (-1)^{xy}
    \langle A_xB_y\rangle .
\end{equation*}
Thus, maximizing the CHSH winning probability is equivalent to maximizing the
corresponding CHSH expression
\begin{equation*}
    I_{\mathrm{CHSH}}
    =
    \sum_{x,y\in\{0,1\}}
    (-1)^{xy}
    \langle A_xB_y\rangle .
\end{equation*}

Let $\mathcal{S}_d$ be a deterministic strategy and let
\[
    f_A(x)=\alpha_x,f_B(y)=\beta_y, \quad \alpha_x,\beta_y\in\{\pm1\}
\]
denote Alice's and Bob's functions. We now introduce loss where the efficiency is $\eta$. Letting $\mathcal S_d$ be the fallback strategy, by~\Cref{eq:lossy_behavior} the correlator $\langle A_x B_y\rangle$ is effectively replaced by 
\begin{equation*}
\begin{aligned}
    \langle A_xB_y\rangle_{\mathcal{S}_d}(\eta)
    \coloneqq
    &\ \eta^2\langle A_xB_y\rangle       +
    \eta(1-\eta)\beta_y\langle A_x\rangle  +
    \eta(1-\eta)\alpha_x\langle B_y\rangle +
    (1-\eta)^2\alpha_x\beta_y .
\end{aligned}
\end{equation*}
Therefore, the Bell expression for the CHSH game is transformed into
\begin{equation*}
    I_{\mathcal{S}_d}(\eta)
    =
    \sum_{x,y\in\{0,1\}}
    (-1)^{xy}
    \langle A_xB_y\rangle_{\mathcal{S}_d}(\eta).
\end{equation*}
and the corresponding winning probability is
\begin{equation*}
    %p_{\mathrm{win}}^{\mathcal{S}_d}(\eta)
    \frac12+\frac18 I_{\mathcal{S}_d}(\eta).
\end{equation*}

Expanding $
    I_{\mathcal{S}_d}(\eta)
$ gives
\begin{equation*}
\begin{aligned}
    I_{\mathcal{S}_d}(\eta)
    =
    &\ \eta^2
    \sum_{x,y\in\{0,1\}}
    (-1)^{xy}\langle A_xB_y\rangle  +
    \eta(1-\eta)
    \Big[
        (\beta_0+\beta_1)\langle A_0\rangle
        +
        (\beta_0-\beta_1)\langle A_1\rangle
    \Big] \\
    &+
    \eta(1-\eta)
    \Big[
        (\alpha_0+\alpha_1)\langle B_0\rangle
        +
        (\alpha_0-\alpha_1)\langle B_1\rangle
    \Big] +
    (1-\eta)^2 \omega(I_\mathrm{CHSH}, \mathcal S_d).
\end{aligned}
\end{equation*}

The $16$ deterministic strategies for the CHSH game split into two classes according to
whether
\[
    \omega(I_{\mathrm{CHSH}},\mathcal S_d)=2
    \qquad\text{or}\qquad
    \omega(I_{\mathrm{CHSH}},\mathcal S_d)=-2.
\]
Each class contains eight strategies. Representatives can be chosen as
\[
    \mathcal S_d^{(+)}:
    \quad
    \langle A_0 \rangle =\langle A_1 \rangle =\langle B_0 \rangle =\langle B_1 \rangle =1,
\]
and
\[
    \mathcal S_d^{(-)}:
    \quad
    \langle A_0 \rangle =\langle A_1 \rangle =1,\qquad \langle B_0 \rangle =\langle B_1 \rangle =-1,
\]
respectively.

The strategies within each class are related by symmetries of the CHSH
expression generated by local input and output relabelings. For example,
the input and output relabelings
\[
    \langle A_0 \rangle \leftrightarrow \langle A_1 \rangle ,\quad \langle B_1 \rangle \mapsto-\langle B_1 \rangle ,
    \qquad\text{and}\qquad
    \langle B_0 \rangle \leftrightarrow \langle B_1 \rangle ,\quad \langle A_1 \rangle \mapsto-\langle A_1 \rangle 
\]
leave
\[
    I_{\mathrm{CHSH}}
    =
    \langle A_0  B_0 \rangle +\langle A_0  B_1 \rangle +\langle A_1 B_0 \rangle -\langle A_1  B_1 \rangle 
\]
unchanged. Together with the simultaneous output flip
$\langle A_x \rangle \mapsto- \langle A_x\rangle $, $\langle B_y\rangle \mapsto-\langle B_y\rangle$, these symmetries connect all eight
strategies within each class. 
It is clear that such relabelings do not change the corresponding $\eta$-lossy value,
Hence it suffices to consider
one representative from each class. We choose $\mathcal S_d^{(+)}$ and $\mathcal S_d^{(-)}$, thereby obtaining the two inequivalent lossy
Bell expressions
\begin{equation*}
    I^{(+)}(\eta)
    =
    \eta^2  \sum_{x,y\in\{0,1\}}
    (-1)^{xy}\langle A_xB_y\rangle
    +
    2(1-\eta)\eta\langle A_0\rangle
    +
    2(1-\eta)\eta\langle B_0\rangle 
    +
    2(1-\eta)^2
\end{equation*}
and
\begin{equation*}
    I^{(-)}(\eta)
    =
    \eta^2  \sum_{x,y\in\{0,1\}}
    (-1)^{xy}\langle A_xB_y\rangle
    +
    2(1-\eta)\eta\langle A_0\rangle
    -
    2(1-\eta)\eta\langle B_0\rangle 
    -
    2(1-\eta)^2.
\end{equation*}
Here we use the shorthand
$$I^{(\pm)}(\eta) \coloneqq I_{\mathcal S_d^{(\pm)}}(\eta).$$

The quantum values of these two Bell expressions have been investigated in Ref.~\cite{gigena2024robustselftestingbellinequalities}: 
\begin{lem}[Theorem 1 in~\cite{gigena2024robustselftestingbellinequalities}]
\label[lem]{lem:gigena}
    Suppose $\eta >0$. Let $\kappa \coloneqq 2\frac{1-\eta}{\eta}$. Define the quartic polynomial
    \begin{equation}
        \label{eq:chsh_polynomial}
            \begin{aligned}
            F_\kappa(\lambda)
            \coloneqq
            &128-64\kappa^2+84\kappa^4-20\kappa^6
            +128\lambda+80\kappa^2\lambda
            +4\kappa^4\lambda-8\kappa^6\lambda \\
            &+16\lambda^2+48\kappa^2\lambda^2
            -11\kappa^4\lambda^2
            -16\lambda^3+4\kappa^2\lambda^3
            -4\lambda^4.
        \end{aligned}
    \end{equation} 
    Then, the quantum value of $I^{(+)}(\eta)$ is 
    \begin{equation*}
        %\label{eq:chsh_lossy_value_opt}
        \omega_q^+(\eta) =
        \begin{cases}
            \max \{\eta^2\lambda^*_{\eta} + 2(1-\eta)^2,2\} & F_\kappa(\lambda) \text{ has a real root}\\
            2 & \text{otherwise},
        \end{cases}
    \end{equation*}
    where $\lambda^*_{\eta}$ is the largest real root of $F_\kappa(\lambda)$.
    And similarly, the quantum value of $I^{(-)}(\eta)$ is 
    \begin{equation*}
        %\label{eq:chsh_lossy_value_non_opt}
        \omega_q^-(\eta) =
        \begin{cases}
            \max \{-\eta^2\lambda'_{\eta} - 2(1-\eta)^2,-4\eta^2+8\eta-2\} & F_\kappa(\lambda)\text{ has a real root}\\
            -4 \eta^2+8\eta-2 & \text{otherwise},
        \end{cases}
    \end{equation*}
    where $\lambda'_{\eta}$ is the smallest real root of $F_\kappa(\lambda)$. 
\end{lem}
Since $F_\kappa(\lambda)$ is a quartic polynomial, both $\lambda^*_{\eta}$ and $\lambda'_{\eta}$ can in principle be solved for analytically. However, their expressions are too complex for comparing $\omega_q^+(\eta)$ and $\omega_q^-(\eta)$ by inspection. In Ref.~\cite{gigena2024robustselftestingbellinequalities}, Gigena et al.\ verified that $\omega_q^+(\eta) \ge \omega_q^-(\eta)$ holds for all $\eta$ numerically. Here we give a rigorous proof of this result.

\begin{thm}\label[thm]{thm:lossy_chsh}
    For all $\eta \in [0,1]$,
    \begin{equation*}
        \omega_q^+(\eta) \ge \omega_q^-(\eta).
    \end{equation*}
    That is, an optimal deterministic strategy for the CHSH game is also an optimal fallback strategy at all $\eta \in [0,1]$.
\end{thm}

\begin{proof}

The theorem statement is trivially true for $\eta=0$.
Moreover, for the CHSH game $\omega_\eta >\omega_c$ only for $\eta>2/3$~\cite{PhysRevA.47.R747,PhysRevA.83.032123}. Therefore, if $\eta<2/3$, we have $\omega_q^+(\eta)  = 2$ (always use the optimal deterministic strategy), and $\omega_q^-(\eta) \leq \omega_c = 2$, and so $\omega_q^+(\eta) \ge \omega_q^-(\eta)$.
Furthermore, if $\eta =1$, there is no loss so the theorem is trivially true.
Hence, if we can prove
$$ \omega_q^+(\eta) \geq \omega_q^-(\eta),$$ 
for $\eta \in (2/3,1)$, or equivalently, $\kappa \coloneqq 2 \frac{1-\eta}{\eta} \in (0,1)$, the theorem follows.

Let $F_\kappa(\lambda)$ be given by~\Cref{eq:chsh_polynomial}. Let
\[
    H(\kappa)
    \coloneqq
    F_\kappa(-2-2\kappa)
    =
    16\kappa^4
    \bigl(\kappa^3-3\kappa^2-8\kappa+4\bigr).
\]
Solving $H(\kappa)=0$ gives
\begin{equation*}
    \kappa=\frac{5\pm\sqrt{17}}{2}, \kappa =0, \kappa=-2.
\end{equation*}

\paragraph{Case 1}
Consider first the case
\[
    \kappa\in
    \left[
        \frac{5-\sqrt{17}}{2},1
        %= 0.438...,1
        % \frac{2}{\sqrt{13}} = 0.554...
    \right).
\]
Let
\[
    \lambda=-2-2\kappa-x,
    \qquad x\ge0.
\]
A direct calculation gives
\begin{equation*}
\begin{aligned}
\partial_\lambda
F_\kappa(-2-2\kappa-x)
=
2\Big[
&8x^3
+(6\kappa^2+48\kappa+24)x^2\\
&+(11\kappa^4+24\kappa^3+72\kappa^2+96\kappa-16)x\\
&-4\kappa^6+22\kappa^5+48\kappa^4
+16\kappa^3+64\kappa^2-32\kappa
\Big].
\end{aligned}
\end{equation*}
The coefficients of $x^3$, $x^2$, $x$, as well as the constant term are positive for $    \kappa\in
    \left[
        \frac{5-\sqrt{17}}{2},
        1
    \right),
$
 thus
we have 
\begin{equation*}
\partial_\lambda
F_\kappa(-2-2\kappa-x) 
>0
\end{equation*}
when $x \ge 0$, or equivalently, $\lambda \le -2-2\kappa$.
Moreover, as 
\begin{equation*}
    H(\kappa)
    \coloneqq
    F_\kappa(-2-2\kappa) \le 0 
\end{equation*}
when $    \kappa\in
    \left[
        \frac{5-\sqrt{17}}{2},
        1
    \right)
$,
$F_\kappa(\lambda)$ has no root on $(-\infty, -2-2\kappa)$.
That is, 
\[
    \lambda'_{\eta}\ge -2-2\kappa.
\]
Hence, by~\Cref{lem:gigena},
$$\omega_q^-(\eta) \leq \eta^2(2+2\kappa )- 2(1-\eta)^2 = -4\eta^2+8\eta-2 \le 2 \leq  \omega_q^+(\eta).$$

\paragraph{Case 2} Now consider the case
\[
    \kappa\in\left(0, \frac{5-\sqrt{17}}{2}\right).
\]
We compute the discriminant of $F_\kappa(\lambda)$~\cite{chatgpt54thinking},
\begin{equation*}
\label{eq:chsh_discriminant}
\mathrm{Disc}_{\lambda}(F_\kappa)
=
-512\,\kappa^2(\kappa-2)^5(\kappa+2)^5(13\kappa^2-4)^3
\bigl(4\kappa^6-11\kappa^4-8\kappa^2-16\bigr).
\end{equation*}
We have $\mathrm{Disc}_{\lambda}(F_\kappa)=0$ only when
$\kappa=0$ or $\kappa=2/\sqrt{13} $ for $\kappa\in[0,1)$.
Note that $\frac{5-\sqrt{17}}{2} 
        = 0.438...
        < \frac{2}{\sqrt{13}} = 0.554...$
Therefore, for $\kappa\in\left(0, \frac{5-\sqrt{17}}{2}\right)$, the terms $\kappa^2$ and $(\kappa+2)^5$ are positive, while $(\kappa-2)^5$, $\bigl(4\kappa^6-11\kappa^4-8\kappa^2-16\bigr)$, and $(13\kappa^2-4)^3$ are negative. 
Hence $\mathrm{Disc}_{\lambda}(F_\kappa)>0$ in this case.
Moreover, for $\kappa \in (0, \frac{5-\sqrt{17}}{2})$,
\[
    F_\kappa(0)
    =
    128-64\kappa^2+84\kappa^4-20\kappa^6>0,
\]
and $F_\kappa(\lambda)\to-\infty$ as $\lambda\to\pm\infty$. Therefore
$F_\kappa(\lambda)=0$ has four real and distinct roots in this case. We denote
these roots by
\[
    \lambda_1(\kappa)>\lambda_2(\kappa)>
    \lambda_3(\kappa)>\lambda_4(\kappa).
\]

We observe that when $
    \kappa \in (0,\frac{5-\sqrt{17}}{2} ),
$
we have $F_\kappa(-2-2\kappa)>0$. Since
$F_\kappa(\lambda)\to-\infty$ as $\lambda\to-\infty$, the intermediate value
theorem guarantees the existence of a real root in
$(-\infty,-2-2\kappa)$. Therefore,
\[
    \lambda_{\eta}'=\lambda_4(\kappa)<-2-2\kappa.
\]
Hence, by~\Cref{lem:gigena}, $\omega_q^-(\eta) = -\eta^2\lambda_4(\kappa) - 2(1-\eta)^2 \le -\eta^2\lambda_4(\kappa)$. Furthermore, $\omega_q^+(\kappa)\geq \eta^2\lambda_1(\kappa) + 2(1-\eta)^2\geq \eta^2\lambda_1(\kappa)$. 
Therefore, to prove
\[
    \omega_q^+(\eta)\ge \omega_q^-(\eta),
\]
it is sufficient to show
\begin{equation*}
%\label{eq:lambda_extreme_sum}
    \lambda_1(\kappa)+\lambda_4(\kappa)\ge 0
\end{equation*}

Define the monic polynomial
\[
    P_\kappa(\lambda)\coloneqq -\frac14 F_\kappa(\lambda).
\]
Then
\begin{equation*}
\begin{aligned}
    P_\kappa(\lambda)
    =
    &\lambda^4+(4-\kappa^2)\lambda^3
    -\left(4+12\kappa^2-\frac{11}{4}\kappa^4\right)\lambda^2 \\
    &-(32+20\kappa^2+\kappa^4-2\kappa^6)\lambda
    -(32-16\kappa^2+21\kappa^4-5\kappa^6).
\end{aligned}
\end{equation*}
First consider $\kappa=0$. We have
\[
    P_0(\lambda)
    =
    \lambda^4+4\lambda^3-4\lambda^2-32\lambda-32
    =
    (\lambda+2)^2(\lambda^2-8).
\]
Hence the smallest and largest roots are
\[
    \lambda_4(0)=-2\sqrt2,
    \qquad
    \lambda_1(0)=2\sqrt2,
\]
so
\[
    \lambda_1(0)+\lambda_4(0)=0.
\]

Now suppose that, for some
\[
    \kappa\in
    \left(0,\frac{5-\sqrt{17}}{2}\right),
\]
one has
\[
    \lambda_1(\kappa)+\lambda_4(\kappa)=0.
\]
Then the extreme roots are $\pm t$ for some $t\ge 0$. Therefore
$P_\kappa$ must factor as
\[
    P_\kappa(\lambda)
    =
    (\lambda^2-t^2)
    \bigl(\lambda^2+(4-\kappa^2)\lambda+m\bigr)
\]
for some real number $m$, because the coefficient of $\lambda^3$ in
$P_\kappa$ is $4-\kappa^2$. Expanding,
\[
    P_\kappa(\lambda)
    =
    \lambda^4+(4-\kappa^2)\lambda^3
    +(m-t^2)\lambda^2
    -(4-\kappa^2)t^2\lambda
    -mt^2 .
\]
Comparing coefficients with the explicit formula for $P_\kappa$, we obtain
\[
    -(4-\kappa^2)t^2
    =
    -(32+20\kappa^2+\kappa^4-2\kappa^6).
\]
Since $4-\kappa^2>0$ for
\[
    \kappa\in
    \left(
        0,\frac{5-\sqrt{17}}{2}
    \right),
\]
we have
\[
    t^2
    =
    \frac{32+20\kappa^2+\kappa^4-2\kappa^6}{4-\kappa^2}
    =
    2\kappa^4+7\kappa^2+8.
\]
Also,
\[
    m-t^2
    =
    -\left(4+12\kappa^2-\frac{11}{4}\kappa^4\right),
\]
and hence
\[
    m
    =
    t^2-\left(4+12\kappa^2-\frac{11}{4}\kappa^4\right)
    =
    \frac{19}{4}\kappa^4-5\kappa^2+4.
\]
Finally, from the constant term,
\[
    -mt^2
    =
    -(32-16\kappa^2+21\kappa^4-5\kappa^6).
\]
Substituting the above expressions for $m$ and $t^2$ yields
\[
    \left(\frac{19}{4}\kappa^4-5\kappa^2+4\right)
    (2\kappa^4+7\kappa^2+8)
    =
    32-16\kappa^2+21\kappa^4-5\kappa^6.
\]
After simplification,
\[
    \kappa^2
    \bigl(38\kappa^6+113\kappa^4-40\kappa^2+16\bigr)=0.
\]
Let $u\coloneqq \kappa^2\ge0$, and define
\[
    R(u)\coloneqq 38u^3+113u^2-40u+16.
\]
It remains to show that $R(u)$ has no positive root. We have
\[
    R'(u)
    =
    114u^2+226u-40
    =
    2(57u^2+113u-20).
\]
Thus the critical points are
\[
    u_{\pm}
    =
    \frac{-113\pm\sqrt{17329}}{114}.
\]
The larger critical point $u_+$ is positive, while $u_-$ is negative. Hence
$R(u)$ has a local maximum at $u_-$ and a local minimum at $u_+$. Moreover,
as
\[
    R(u_+)>0,
\]
$R(u)$ has no positive root. This implies
\[
    \lambda_1(\kappa)+\lambda_4(\kappa)\neq 0
\]
for
\[
    \kappa\in
    \left(0,\frac{5-\sqrt{17}}{2}\right).
\]

Since the coefficients of $P_\kappa(\lambda)$ depend continuously on $\kappa$,
all four roots $\{\lambda_i(\kappa)\}_{i=1}^4$ are real continuous functions of
$\kappa$ on
\[
    \left(0,\frac{5-\sqrt{17}}{2}\right),
\]
according to Theorem~5.1 in Ref.~\cite{kato1966perturbation}. Therefore, the
largest and smallest roots are continuous functions of $\kappa$, and  $\lambda_1(\kappa)+\lambda_4(\kappa)$ varies continuously on this
interval. We can simply set
\[
    \kappa=\frac14\in
    \left(0,\frac{5-\sqrt{17}}{2}\right),
\]
and one finds that $P_{1/4}(\lambda)=0$ has real roots
\[
    \lambda_1\approx 2.904,
    \qquad
    \lambda_4\approx -2.802.
\]
Their sum is positive. Thus, by the continuity of
$\lambda_1(\kappa)$ and $\lambda_4(\kappa)$, 
we obtain
\[
    \lambda_1(\kappa)+\lambda_4(\kappa) > 0
\]
for
\[
    \kappa\in
    \left(
        0,\frac{5-\sqrt{17}}{2}
    \right).
\] 

Consequently, we have proved that $\omega_q^+(\eta) \ge \omega_q^-(\eta)$ for all $\eta \in [0,1]$, and the theorem follows.

\end{proof}

\subsection{Numerical studies}
The mathematical argument developed for the CHSH inequality in~\Cref{thm:lossy_chsh} is rather ad hoc and does not extend directly to general Bell inequalities. We therefore study several representative Bell inequalities under loss with numerical methods. For each deterministic strategy, we use the NPA hierarchy~\cite{navascues2008convergent} to obtain an upper bound on the corresponding $\eta$-lossy value and employ see-saw optimization~\cite{werner2001bell,Ketjl} to obtain lower bounds.
All numerical results in this section are obtained using the NPA hierarchy
at level $3$. For the see-saw optimization, we set the local Hilbert-space
dimension to be equal to the number of outputs per party unless otherwise specified. At each value of $\eta$, we run the see-saw optimization with random initializations repeatedly 
and keep the largest value found. 

These numerical bounds also allow us to estimate the threshold efficiency
$\eta^*$. Given the numerical precision of our see-saw and NPA calculations,
we regard a value exceeding the classical bound by more than $10^{-8}$ as
evidence of quantum advantage. Differences of magnitude at
most $10^{-8}$ are treated as numerically inconclusive. We take the smallest sampled efficiency at
which the see-saw lower bound shows a quantum advantage as an upper bound on
$\eta^*$, and the largest sampled efficiency at which the NPA upper bound
shows no quantum advantage as a lower bound. Combining these
two criteria yields an interval containing $\eta^*$.

We consider several well-known Bell inequalities, including the $I_{3322}$ inequality~\cite{collins2004relevant}, CGLMP inequality~\cite{collins2002bell}, and the chained inequality~\cite{braunstein1990wringing}.
For each inequality, we introduce loss and study how the optimal fallback strategy $\mathcal S_d^*$   
and the lossy value $\omega_\eta$ depend on the detection efficiency $\eta$.

Before discussing individual examples in detail, we summarize a pattern observed
in these numerical studies. We find the inequalities that we study can be split into two classes, distinguished by how the optimal fallback strategies vary with $\eta$.
In some cases, there exists a single deterministic strategy $\mathcal S_d$ which is an optimal fallback strategy at all $\eta \in [0,1]$. 
We call this a \textbf{universal fallback strategy}. In particular, since $\mathcal S_d$ is an optimal fallback strategy at $\eta=0$, it is also an optimal deterministic strategy.
As we proved in~\Cref{thm:lossy_chsh}, the CHSH game has a universal fallback strategy. 
Similar to the CHSH case, throughout this section, deterministic strategies related by input and
output relabelings that leave the Bell expression invariant are regarded as equivalent, as they yield the same
lossy value. Accordingly, when the optimal strategy identified
numerically varies with $\eta$ only within the same equivalence class,
we refer to any representative of this class as a universal fallback
strategy.

In our studies, we find numerical evidence that the CGLMP inequality with
three outputs per party and the chained inequality with three inputs per
party each has a universal fallback strategy. However, not all inequalities admit a universal fallback strategy. The $I_{3322}$ inequality is such an example, as summarized in Table~\ref{tab:classification}.

\begin{table}[t]
\centering
\caption{Classification of nonlocal games according to whether they have a universal fallback strategy. The game $I_{\mathrm{corr}22}$ is randomly generated and defined in~\Cref{subsubsec:random}. }
\label{tab:classification}
\begin{tabular}{cc}
\toprule
Class & Examples \\
\midrule
Has universal fallback &
CHSH, CGLMP (3 outputs), Chained (3 inputs)
\\
No universal fallback  &
$I_{3322}, I_{\mathrm{corr}22}$
\\
\bottomrule
\end{tabular}
\end{table}

\subsubsection{The CGLMP inequality} 
As a first example, we consider the CGLMP inequality with three outputs per party~\cite{collins2002bell}. This inequality was studied from the perspective of the detection loophole in~\cite{Massar2002}, but in their analysis they treated a loss event as an additional output possibility. Here we assume the parties must return an output from the original alphabet.

Let $a_x, b_y \in \{0,1,2\}$ respectively be Alice's and
Bob's outputs for input $x,y \in \{1,2\}$. 
Define the probability $ P( a_x = b_y +k )$ as 
    \begin{equation*}
        P(a_x=b_y+k) \equiv \sum_{j=0}^{2} P(a_x=j,b_y=j+k \mod 3).
    \end{equation*}
    The three-output CGLMP inequality is 
    \begin{equation*}
    \begin{aligned}
                I_3^{\mathrm{CGLMP}} = &
                \bigg( \big[ P(a_1=b_1) + P(b_1=a_2+1) + P(a_2=b_2) + P(b_2=a_1)\big] \\
                &- \big[ P(a_1=b_1-1) + P(b_1=a_2) + P(a_2=b_2-1) + P(b_2=a_1-1)\big] \bigg) \le 2.
    \end{aligned}
    \end{equation*}
    
Numerical results indicate that the threshold efficiency lies within
$\eta^* \in [0.67,0.70]$. For visualization we scan $\eta \in [0.6,1]$ in
increments of $0.01$, enumerate all $3^4=81$ deterministic strategies,
and evaluate the corresponding $\eta$-lossy values using the NPA
hierarchy and see-saw optimization.
The resulting upper and lower
bounds are shown in \Cref{fig:cglmp_lossy_value}.
\begin{figure}[h!]
    \centering
    \includegraphics[width=0.5\linewidth]{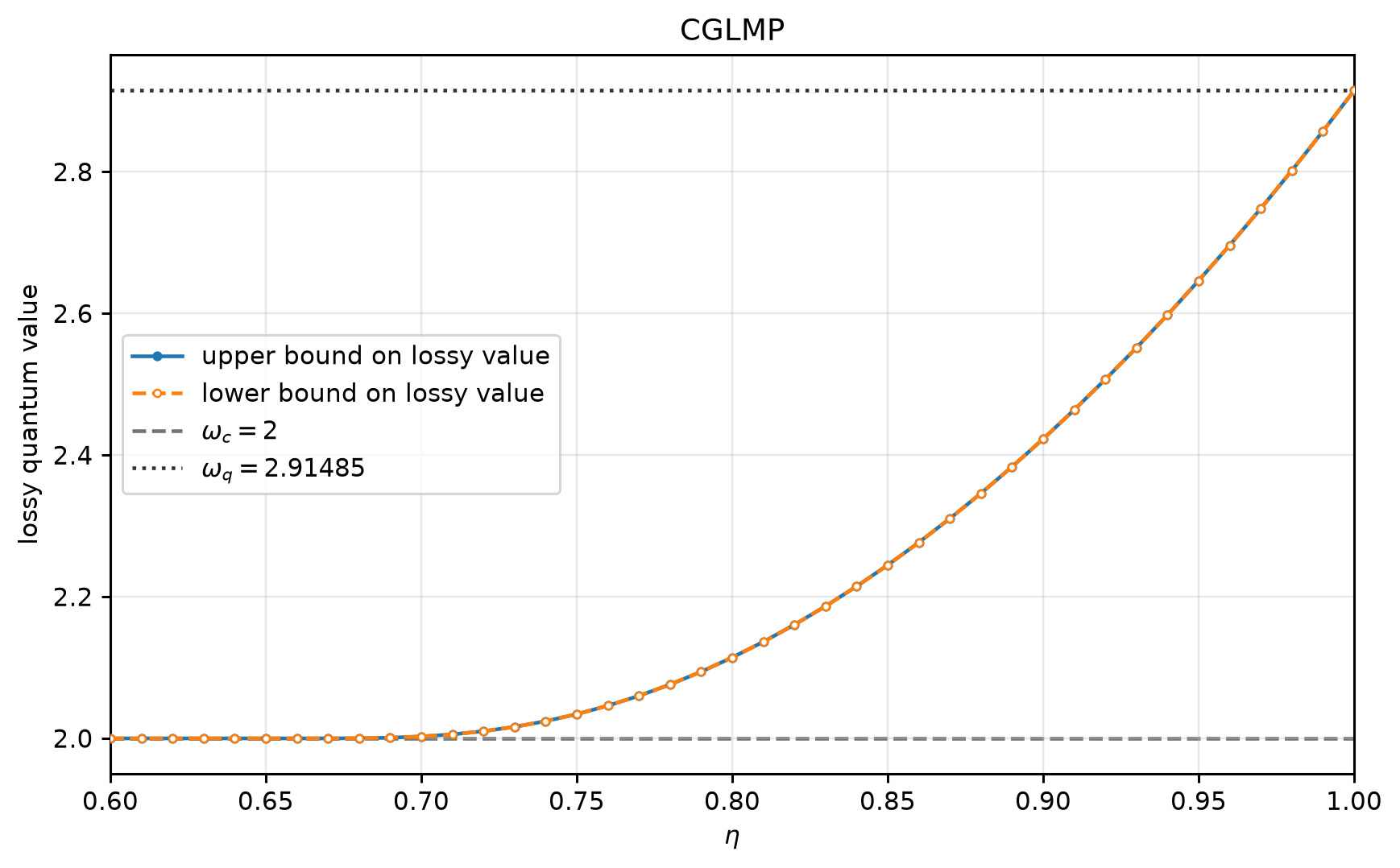}
    \caption{$\eta$-lossy value of the three-output CGLMP inequality.} 
    \label{fig:cglmp_lossy_value}
\end{figure}

We find that the fallback strategy
\begin{equation*}
    \mathcal{S}_d^{\mathrm{opt}} =
    \begin{cases}
        f_1(1)=2,\quad f_1(2)=0\\
        f_2(1)=1,\quad f_2(2)=2
    \end{cases}
\end{equation*}
is optimal at all sampled values of $\eta$. 
Moreover, $\mathcal{S}_d^{\mathrm{opt}}$ is an optimal deterministic
strategy for the CGLMP inequality. It therefore remains an
optimal fallback strategy in the regime $\eta\in[0,0.6]$, below the threshold efficiency where the lossy value equals the classical value.

To achieve a better lower bound, we used the see-saw method with local dimensions $d_A=d_B=3$. Starting from $\eta=1$, we evaluated the sampled values of $\eta$ in descending order toward $\eta=0$. At $\eta=1$, the see-saw was initialized using the known optimal CGLMP state and measurements; at each subsequent point, the optimized solution from the neighboring higher-$\eta$ point was used as a warm start. 
A worst-case gap of order $10^{-6}$ remains between the see-saw lower and NPA upper bounds on the lossy value. 
Nevertheless,
this gap does not affect the identification of the optimal fallback
strategy: the see-saw lower bound on the $\eta$-lossy value with respect to  
$\mathcal{S}_d^{\mathrm{opt}}$ is strictly larger than the NPA upper
bounds with respect to all other inequivalent fallback
strategies. 

Together with the results for $\eta\in[0,0.6]$, our numerical results
thus suggest that the three-output CGLMP inequality has a universal
optimal fallback strategy.

\subsubsection{The chained inequality}
As a second example, we consider the chained Bell inequality~\cite{braunstein1990wringing} with three
inputs per party.
\begin{equation*}
        \langle I_3^{\mathrm{chained}} \rangle = \sum_{i=1}^3 ( \langle A_iB_i\rangle  + \langle A_{i+1}B_{i} \rangle) \le 4,
\end{equation*}
where the last term is $-\langle A_1 B_3\rangle$. 
Numerical results indicate that the threshold efficiency lies within
$\eta^* \in [0.80,0.81]$. For visualization we scan
$\eta\in[0.65,1]$ in increments of $0.01$, enumerate all
$2^6=64$ deterministic strategies, and evaluate the corresponding
$\eta$-lossy values using the NPA hierarchy and see-saw optimization.
The results are shown in \Cref{fig:chained_lossy_value}.

We find that the fallback strategy
\begin{equation*}
    \mathcal{S}_d^{\mathrm{opt}}
    =
    \begin{cases}
        A_1=1, \quad A_2=1, \quad A_3=1\\
        B_1=1, \quad B_2=1, \quad B_3=1
    \end{cases}.
\end{equation*}
is optimal at all sampled values of $\eta$.
Moreover, $\mathcal{S}_d^{\mathrm{opt}}$ is an optimal deterministic
strategy for the original chained Bell inequality. It therefore remains
an optimal fallback strategy in the regime $\eta\in[0,2/3]$, where no
quantum advantage is possible. Taken together, our numerical results
suggest that $I_3^{\mathrm{chained}}$ has a universal optimal
fallback strategy.
\begin{figure}[h!]
    \centering
    \includegraphics[width=0.5\linewidth]{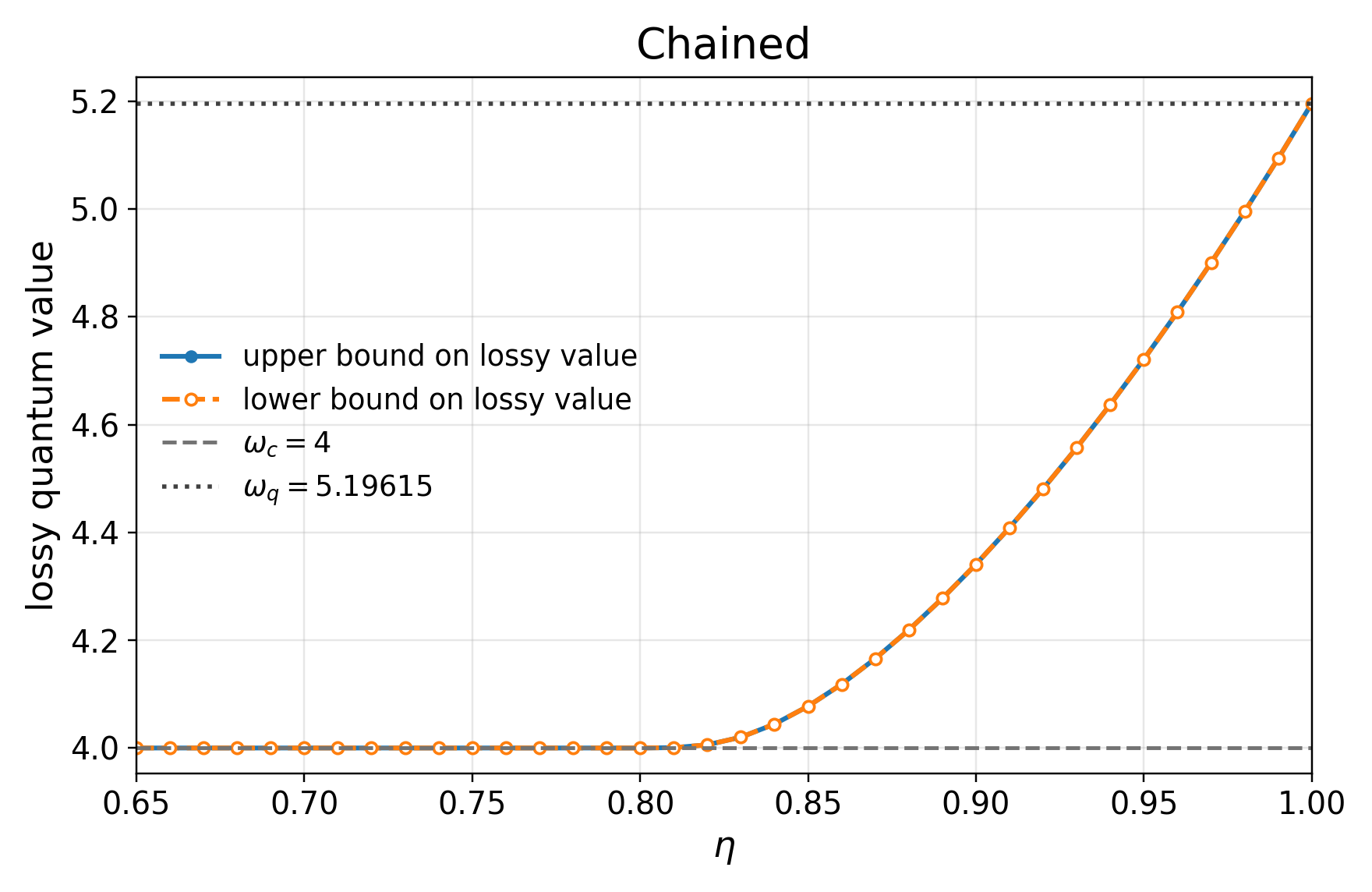}
    \caption{$\eta$-lossy value of the three-input chained inequality. }
    \label{fig:chained_lossy_value}
\end{figure}

\subsubsection{The $I_{3322}$ inequality}
As a third example, we consider the $I_{3322}$ inequality~\cite{collins2004relevant}, which has three inputs and two outputs per party:
\begin{equation*}
\begin{aligned}
I_{3322}
=&\ p(00|00)+p(00|01)+p(00|02)
+p(00|10)+p(00|11)-p(00|12) \\
&+p(00|20)-p(00|21)
-2p_A(0|0)-p_A(0|1)-p_B(0|0)
\le 0.
\end{aligned}
\end{equation*}
Here, $p_A$ and $p_B$ are the marginal distributions for Alice and Bob, respectively.

For $I_{3322}$, 
numerical results indicate that the threshold efficiency lies within
$\eta^* \in [0.66,0.67]$. We therefore scan $\eta\in[0.65,1]$ in increments
of $0.01$, enumerate all $2^6=64$ deterministic strategies, and evaluate
the corresponding $\eta$-lossy values using the NPA hierarchy and
see-saw optimization, as shown in \Cref{fig:i3322_lossy_value}. 
\begin{figure}
    \centering
    \includegraphics[width=0.5\linewidth]{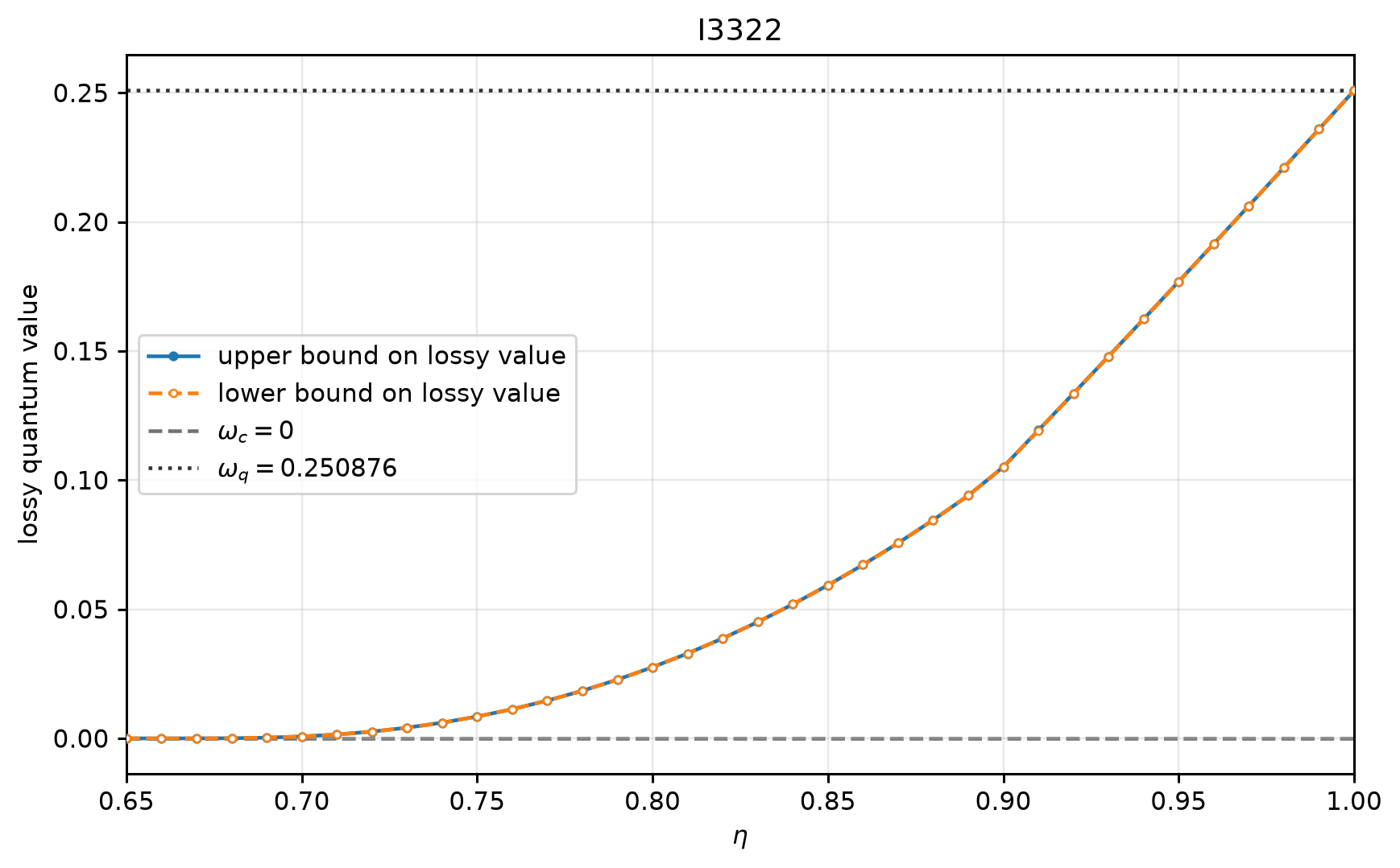}
    \caption{The $\eta$-lossy value of the $I_{3322}$ inequality. }
    \label{fig:i3322_lossy_value}
\end{figure}
The maximal quantum violation of $I_{3322}$ is known to depend sensitively on the
Hilbert-space dimension. We employed a structured projective see-saw
with local dimension $d=51$, supplemented by higher-dimensional
refinements at some points, for example, at $\eta=0.90$ ($d=79$) and $\eta=1$ ($d=149$).
For each candidate fallback strategy, we initialized the projective see-saw using the Pal--V\'ertesi construction~\cite{pal2010maximal}. When available, the optimized state and measurements from the neighboring higher-$\eta$ point were used as an additional warm start. The largest feasible value over all initializations was retained as
the see-saw lower bound. The worst-case gap between this lower bound
and the corresponding NPA upper bound is on the order of $10^{-5}$.

Our numerical results establish that $I_{3322}$ does not have a universal fallback strategy.
Consider the
deterministic strategy
\begin{equation*}
\mathcal S_d:
\quad
\begin{cases}
f_A(0)=1,\quad f_A(1)=1,\quad f_A(2)=0,\\
f_B(0)=0,\quad f_B(1)=0,\quad f_B(2)=0,
\end{cases}
\end{equation*}
which is a non-optimal deterministic strategy for the $I_{3322}$ inequality, as it
gives $I_{3322}=-1<0$.

Among the $2^6=64$ deterministic strategies, $20$ are optimal for the
original $I_{3322}$ inequality, saturating the classical bound
$I_{3322}=0$. For all sampled values of $\eta\in[0.67,0.99]$, the
largest $\eta$-lossy value among these $20$ strategies is attained by the fallback strategy
\begin{equation*}
\mathcal S_d^{\mathrm{opt}}:
\quad
\begin{cases}
f_A(0)=0,\quad f_A(1)=1,\quad f_A(2)=0,\\
f_B(0)=0,\quad f_B(1)=0,\quad f_B(2)=0.
\end{cases}
\end{equation*}

As shown in \Cref{fig:i3322_compare}, for $\eta\ge 0.91$, the
see-saw lower bound obtained with $\mathcal S_d$ is
larger than the NPA upper bound obtained by maximizing over all optimal deterministic strategies. Therefore, in this regime,
$\mathcal S_d$ outperforms all optimal deterministic
strategies for the $I_{3322}$ inequality.
Since $\mathcal S_d$ is itself a non-optimal deterministic strategy, our numerical results show that $I_{3322}$ does not have 
a universal optimal fallback strategy.

\begin{figure}
    \centering
    \includegraphics[width=0.5\linewidth]{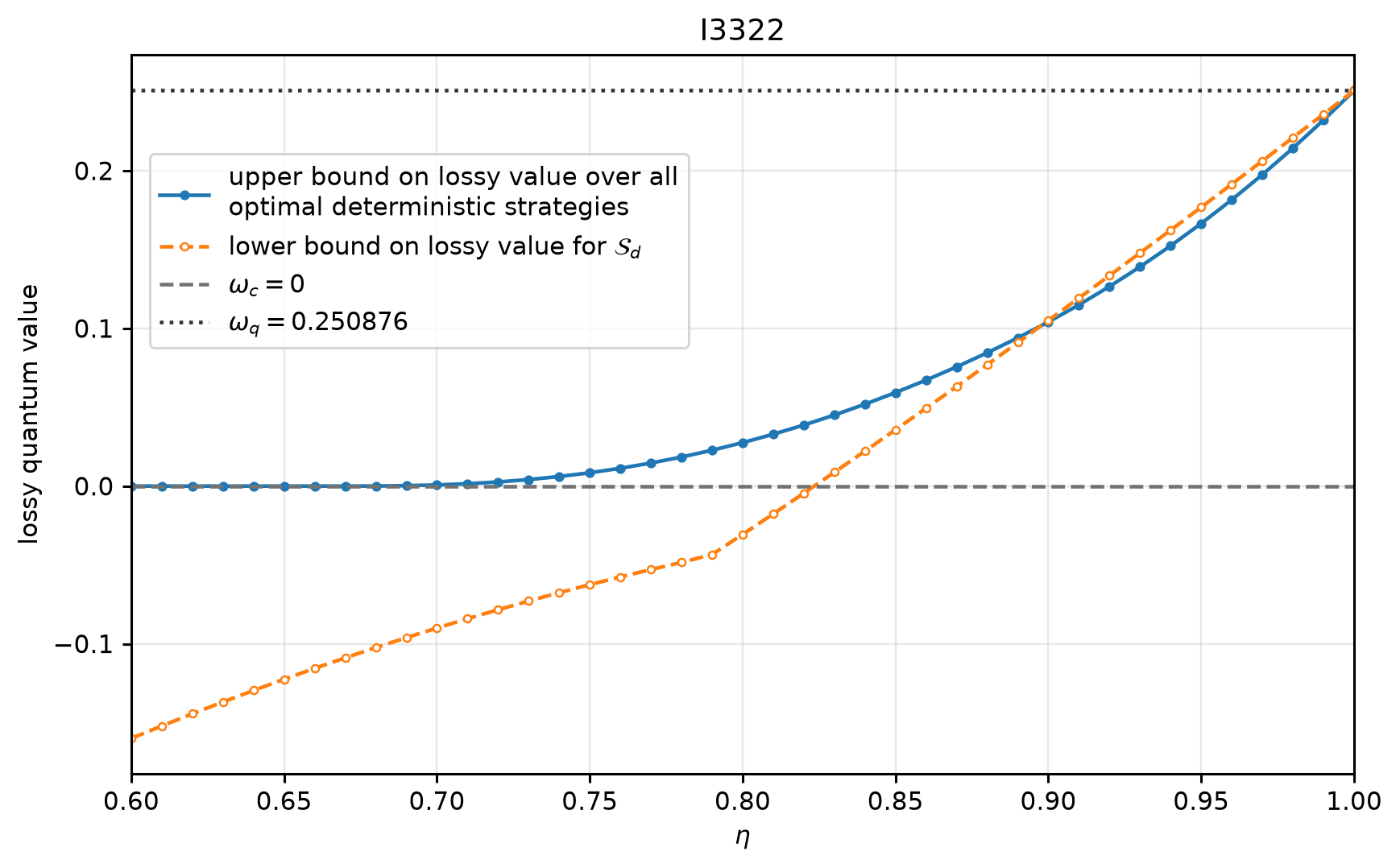}
    \caption{Comparison of fallback strategies for the
    \(I_{3322}\) inequality. The blue solid curve represents the NPA upper bound maximized over all optimal deterministic strategies. For
    \(\mathcal S_d\), which is a non-optimal deterministic strategy, we plot with the orange dashed curve the see-saw lower bound. For the sampled points on $\eta \in [0.92,0.99]$, we find the non-optimal deterministic strategy $\mathcal S_d$ outperforms all optimal deterministic strategies. }  
    \label{fig:i3322_compare}
\end{figure}

\subsubsection{Random game $I_{\mathrm{corr}22}$}
\label{subsubsec:random}
We wish to find other examples of nonlocal games that have no universal fallback strategy. The above examples were selected from well-studied Bell inequalities with a lot of structure. It may be more promising to try less structured examples.

To do this, we randomly generate bipartite Bell expressions with
two inputs and two outputs per party. For each generated expression, we enumerate
all \(2^4=16\) deterministic strategies and compute the corresponding
$\eta$-lossy values.

Consider the example 
\begin{equation*}
    I_{\mathrm{corr22}} = \frac{1}{2}\langle A_0\rangle+ \frac{1}{2}\langle A_1\rangle- \frac{1}{2}\langle B_0\rangle + \langle A_0B_0\rangle + \langle A_0B_1\rangle + \langle A_1B_0\rangle- \langle A_1B_1\rangle.
\end{equation*}
The classical value is $5/2$. The $\eta$-lossy value of \(I_{\mathrm{corr22}}\), computed using the NPA
hierarchy and see-saw optimization for $\eta \in [0.84,1]$ in increments of $0.01$, is shown in
\Cref{fig:corr22_lossy_value}. The threshold efficiency is $\eta^* \in [0.87,0.88]$.
\begin{figure}
    \centering
    \includegraphics[width=0.5\linewidth]{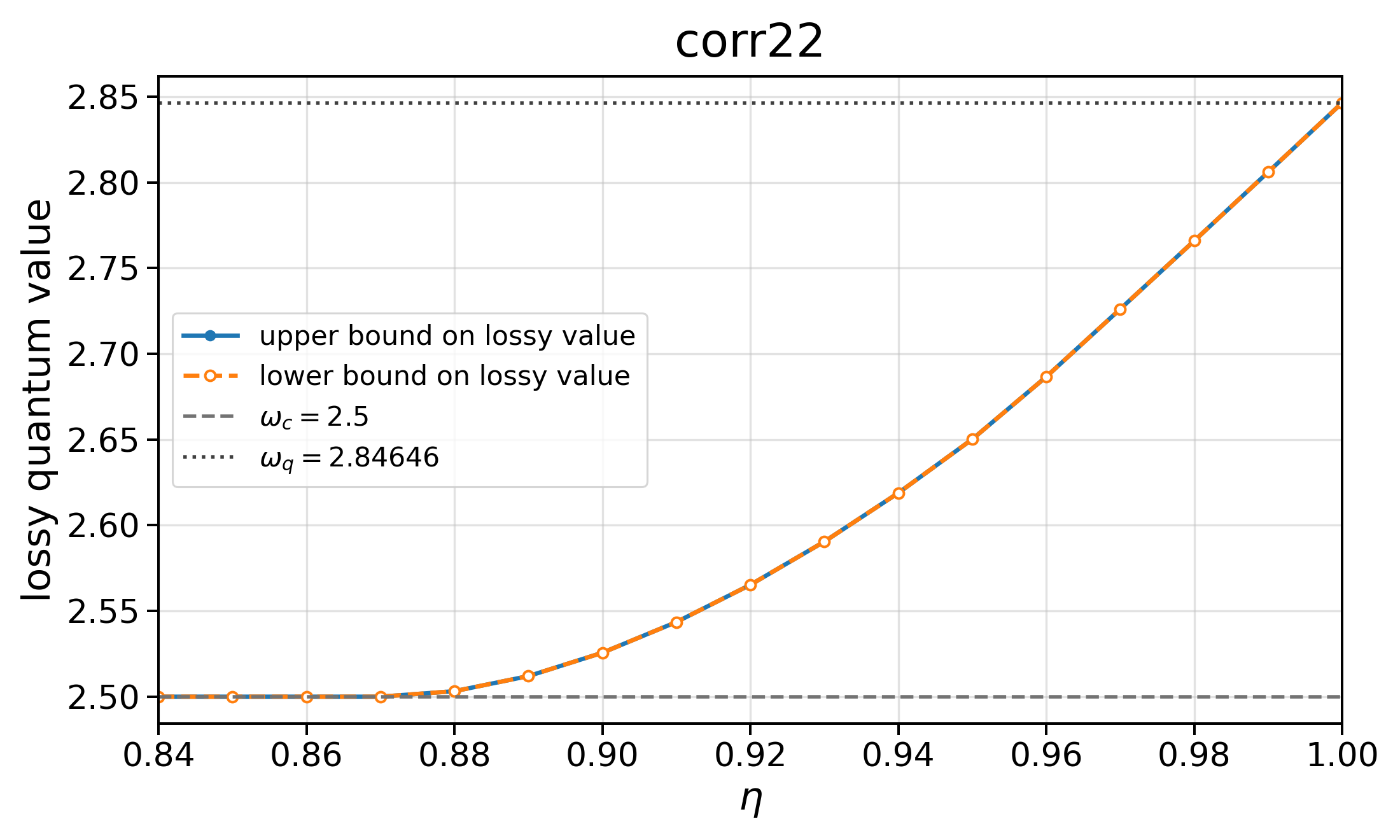}
    \caption{$\eta$-lossy value of \(I_{\mathrm{corr22}}\). }
    \label{fig:corr22_lossy_value}
\end{figure}
\begin{figure}
    \centering
    \begin{subfigure}[t]{0.45\linewidth}
        \centering
        \includegraphics[width=\linewidth]{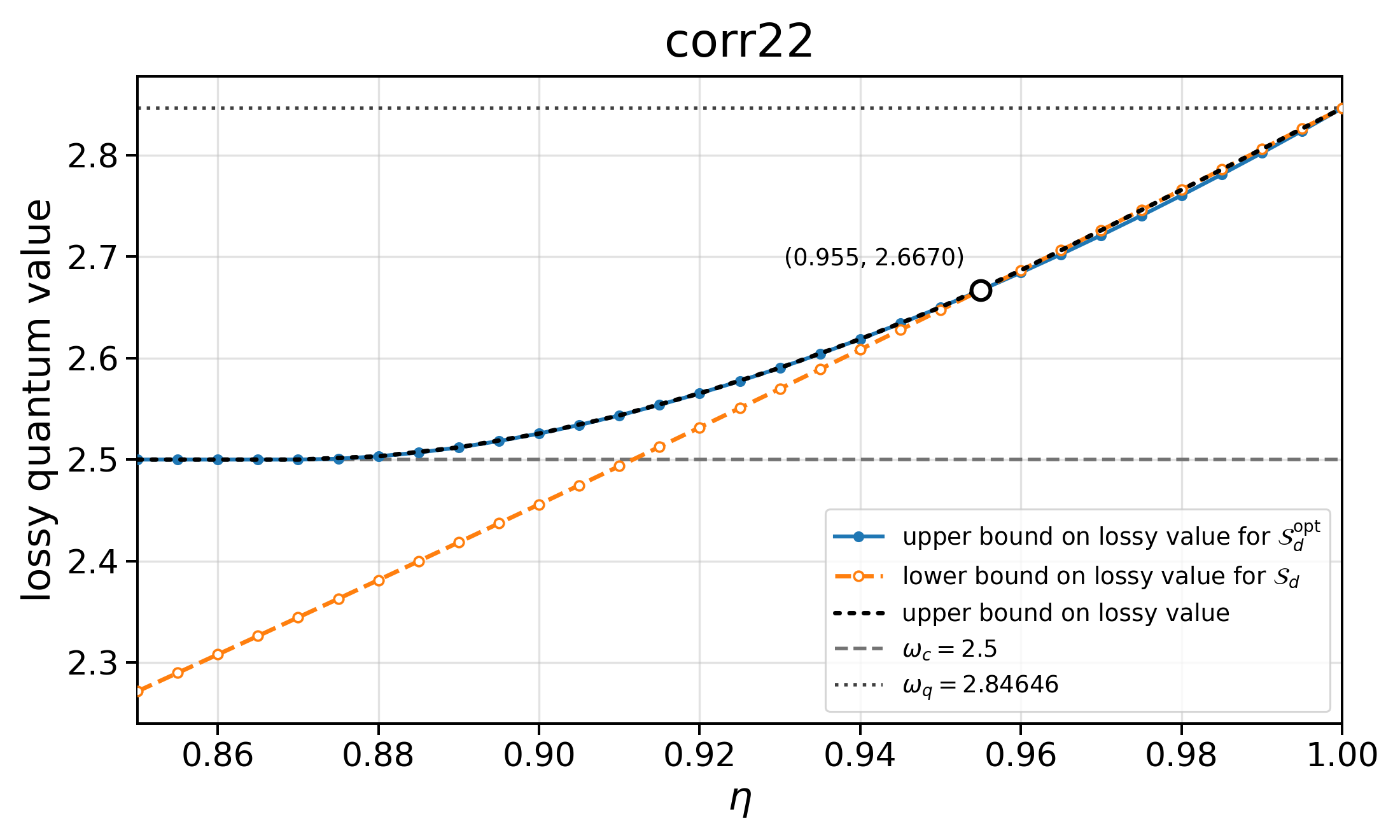}
        \caption{}
        \label{fig:corr22_compare_main}
    \end{subfigure}
    \hfill
    \begin{subfigure}[t]{0.5\linewidth}
        \centering
        \includegraphics[width=\linewidth]{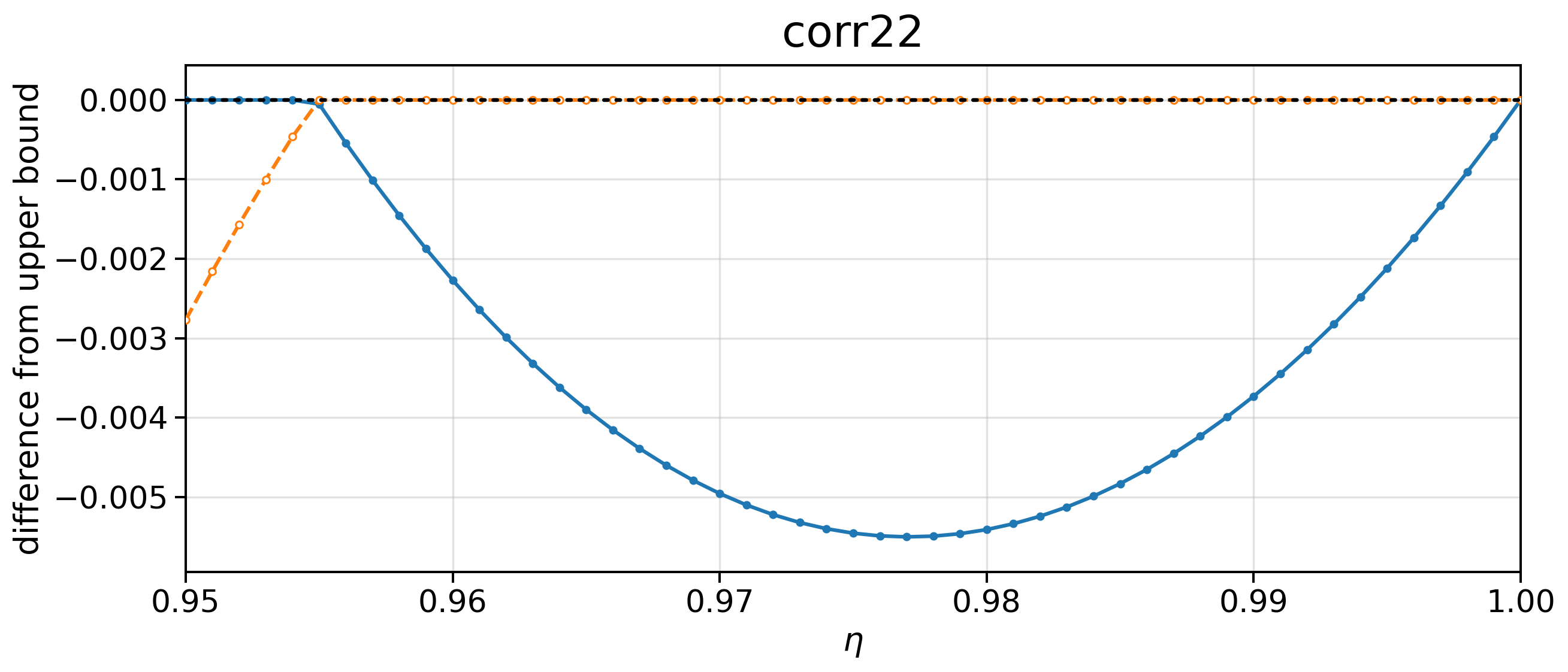}
        \caption{}
        \label{fig:corr22_compare_diff}
    \end{subfigure}
    \caption{Comparison of the $\eta$-lossy value for
    $I_{\mathrm{corr22}}$ with respect to different fallback 
    strategies. (a) The black curve shows the NPA upper bound on the
    overall $\eta$-lossy value, while the blue curve shows the NPA
    upper bound on the $\eta$-lossy value maximized over all optimal
    deterministic strategies. The orange curve shows a lower bound
    on the $\eta$-lossy value with respect to $\mathcal{S}_d$. 
    (b) For $\eta\in[0.95,1]$ with step size $0.001$, the differences between the overall
    NPA upper bound and, respectively, the NPA upper bound over all
    optimal deterministic strategies and the lower bound obtained
    with $\mathcal S_d$ are shown. }
    \label{fig:corr22_compare}
\end{figure}

Consider the non-optimal deterministic strategy
\begin{equation*}
    \mathcal S_d:
    \quad
    A_0=1,\quad
    A_1=1,\quad
    B_0=-1,\quad
    B_1=-1.
\end{equation*}
This obtains a value of $-\frac 1 2 < \frac 5 2$. Meanwhile, 
there are four optimal deterministic strategies for
$I_{\mathrm{corr22}}$. Among them, we find numerically that, for all
sampled values of $\eta$, the largest $\eta$-lossy value is attained by
\begin{equation*}
    \mathcal S_d^{\mathrm{opt}}:
    \quad
    A_0=1,\quad
    A_1=1,\quad
    B_0=1,\quad
    B_1=1.
\end{equation*}

We compare the $\eta$-lossy value with respect to $\mathcal S_d$ to that of $\mathcal{S}_d^{\mathrm{opt}}$ in \Cref{fig:corr22_compare}. We observe that when $\eta \ge 0.955$, the non-optimal deterministic strategy $\mathcal S_{d}$ leads to larger $\eta$-lossy value. Moreover, within numerical precision, the fallback strategy
\(\mathcal S_d\) leads to the overall $\eta$-lossy value
\(\omega_\eta\) in this interval. We conclude that $I_{\mathrm{corr}22}$ does not have a universal optimal fallback strategy.

\subsubsection{Optimal fallback strategy near the threshold efficiency}
Another important question is whether an optimal fallback strategy near the threshold efficiency $\eta^*$ is always an optimal deterministic strategy. 
Surprisingly, we find numerical evidence that this is not true in general.
As an example, consider the Bell expression 
\begin{equation*}
\begin{aligned}
J_{3322}=&
\langle A_0B_0\rangle+\langle A_0B_1\rangle+\langle A_1B_0\rangle-\frac13\langle A_1B_1\rangle+\frac23\langle A_1B_2\rangle+\frac23\langle A_2B_1\rangle-\frac23\langle A_2B_2\rangle\\
&-\frac14\langle A_0\rangle-\frac1{15}\langle A_1\rangle-\frac18\langle A_2\rangle+\frac1{16}\langle B_0\rangle-\frac1{10}\langle B_1\rangle+\frac3{10}\langle B_2\rangle.
\end{aligned}
\end{equation*}
The classical value is $\omega_c = 329/80= 4.1125$. 

The strategy 
\begin{equation*}
\mathcal S_d^\mathrm{opt}:
\quad
\begin{cases}
A_0=-1,\quad A_1=-1,\quad A_2=-1,\\
B_0=-1,\quad B_1=-1,\quad B_2=1,
\end{cases}
\end{equation*}
is the unique optimal deterministic strategy for \(J_{3322}\). And let
\begin{equation*}
\mathcal S_d:
\quad
\begin{cases}
A_0=-1,\quad A_1=1,\quad A_2=-1,\\
B_0=-1,\quad B_1=-1,\quad B_2=1,
\end{cases}
\end{equation*}
be a non-optimal deterministic strategy which achieves $\frac{191}{48} < \frac{329}{80}$.
We compute the threshold efficiency to lie in $\eta^* \in [0.928,0.929]$.
We compute a lower bound on the $\eta$-lossy value with respect to $\mathcal S_d$ using the see-saw method for $\eta \in [0.92,1]$ using step size $0.001$. We next compute an upper bound on the $\eta$-lossy value with respect to $\mathcal S_d^\mathrm{opt}$ using the NPA hierarchy. Moreover, an upper bound on the overall $\eta$-lossy 
value is also calculated via the NPA hierarchy. 
The result is shown in~\Cref{fig:J3322_lossy_value}.  
\begin{figure}[h!]
    \centering
    \includegraphics[width=0.5\linewidth]{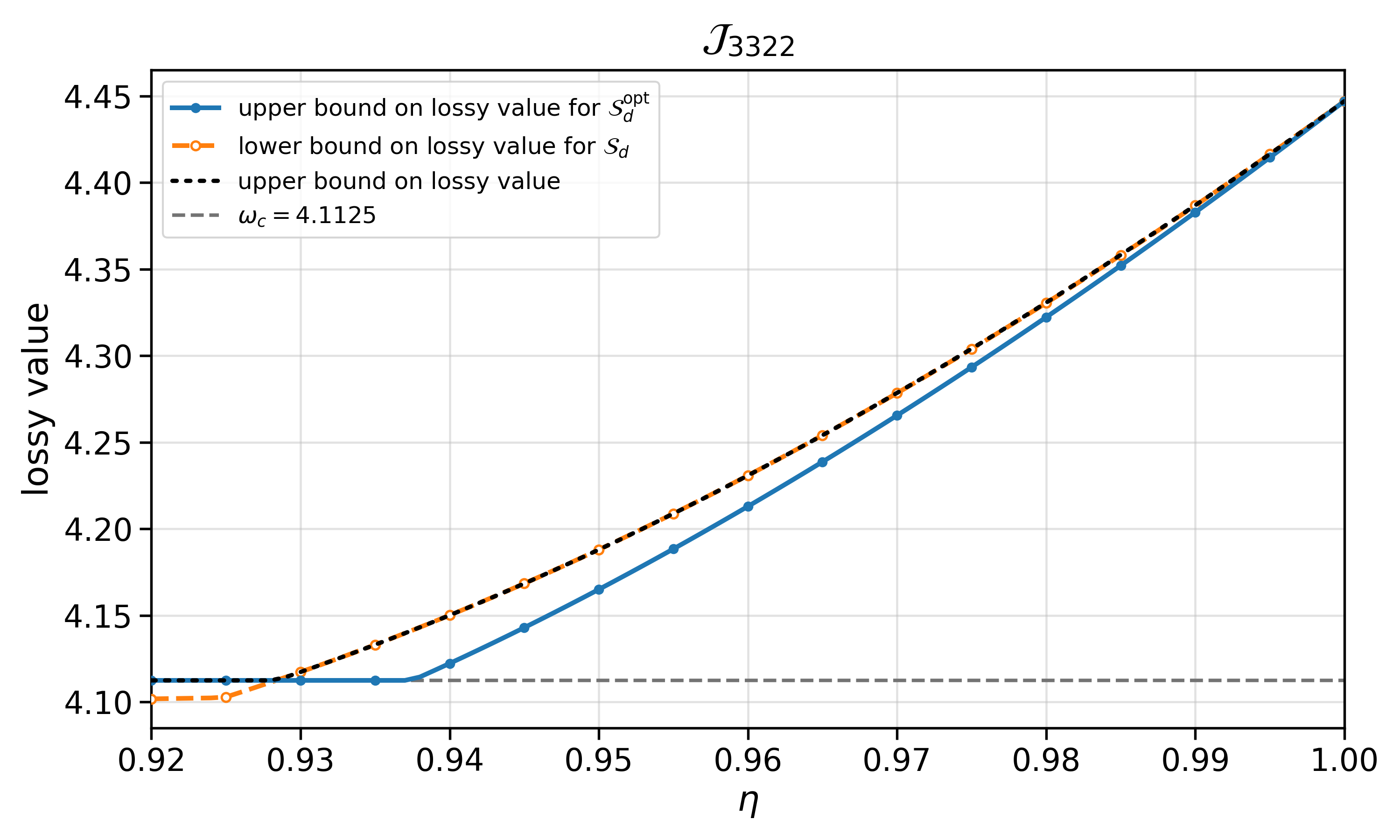}
    \caption{Comparison of the $\eta$-lossy values obtained via NPA optimization
over all deterministic strategies, via NPA optimization for the unique
optimal deterministic strategy $\mathcal S_d^{\mathrm{opt}}$, and via
see-saw optimization for another non-optimal deterministic strategy
$\mathcal S_d$.}
    \label{fig:J3322_lossy_value}
\end{figure}

It is observed that when $\eta > \eta^*$, the upper bound on the overall $\eta$-lossy value coincides with the lower bound on the $\eta$-lossy value with respect to $\mathcal S_d$, which is larger than the $\eta$-lossy value with respect to $\mathcal S_d^\mathrm{opt}$. 
Therefore, we observe that there exist Bell inequalities for which the
optimal fallback strategy at $\eta$ such that $\omega_\eta > \omega_c$ is not an optimal deterministic strategy.

This finding challenges the intuition that the optimal
deterministic strategy should also be
an optimal fallback strategy near $\eta$ where $\omega_\eta \approx \omega_c$. Instead, it suggests that non-optimal deterministic strategies perform better as fallback strategies in the presence of loss and can even help decrease the critical threshold. 
For example, if we use $\mathcal S_d^\mathrm{opt}$ as the fallback strategy, the quantum advantage already disappears when $\eta \le 0.937$, larger than the threshold efficiency $\eta^*$ obtained from $\mathcal S_d$.
This is analogous to the finding that non-maximally entangled states can achieve higher values for the CHSH game in the presence of loss~\cite{PhysRevA.47.R747}.
Our observation has important implications for loophole-free Bell tests
and real-world applications of Bell inequalities~\cite{ding2024coordinating,ding2026quantum}.

\section{Conclusion}
\label{sec:conc}

How to achieve the $\eta$-lossy value for a Bell inequality is an important question for closing the detection loophole and also for applications of Bell inequality violations if the physical entanglement distribution network suffers from loss. 
To achieve the $\eta$-lossy value, the parties must choose a fallback strategy, which is a deterministic strategy they use when loss occurs. In this paper, we systematically study how to choose the optimal fallback strategy for a general Bell inequality at an efficiency $\eta$ using a combination of analytic and numerical tools. Answering this question can decrease the minimum efficiency required to close the detection loophole or to realize the applications of Bell inequality violation. 
These are important questions in this nascent era of quantum networks where entanglement distribution is largely limited by photon loss in transmission media.

We list here some future directions. The most natural next step is to extend our results to the multipartite and non-symmetric loss settings. This will be useful for applications of Bell inequality violation such as quantum telepathy and DIQKD where realistic physical simulations need to be conducted, including the effect of loss.
For this purpose it would also be important to develop numerical techniques that can compute lossy values more efficiently than simply iterating over all possible fallback strategies and analyzing the corresponding lossy Bell expression. In particular, it would be useful to find a good heuristic for the optimal fallback strategy.
Another extension is to consider latency-constrained (LC) games~\cite{ding2025quantum}, where a subset of parties can communicate within a specified latency constraint. In particular, in this setting we have the additional interesting feature that \emph{some parties can tell one another that a loss occurred}. The receiving parties can then adjust their quantum or fallback strategies accordingly, leading to new, adaptive strategies. 

\paragraph{Acknowledgments} DD would like to thank God for all of His provisions. 

\paragraph{AI disclosure} ChatGPT 5.6-thinking was used to assist in the proof of a universal fallback strategy for the CHSH inequality and in finding the $J_{3322}$
counterexample. All AI-assisted
arguments and results were independently verified by the authors.

\bibliographystyle{unsrt}
\bibliography{ref}
\end{document}